\documentclass{lmcs}

\usepackage[UKenglish]{babel}
\usepackage{graphicx}
\usepackage{mathtools}
\usepackage{amsmath,amssymb,amsfonts}
\usepackage{caption}
\usepackage{subcaption}
\usepackage{xspace}
\usepackage{xcolor}
\AtEndPreamble{\RequirePackage[nameinlink,noabbrev,capitalize]{cleveref}}
\usetikzlibrary{arrows,shapes,automata, calc, chains, matrix, trees, positioning, scopes,decorations.pathreplacing}
\usetikzlibrary{decorations.pathmorphing}
\tikzset{snake it/.style={decorate, decoration=snake}}

\newtheorem{theorem}[thm]{Theorem}
\newtheorem{lemma}[thm]{Lemma}
\newtheorem{proposition}[thm]{Proposition}
\theoremstyle{definition}
\newtheorem{example}[thm]{Example}
\newtheorem{remark}[thm]{Remark}

\newcommand{\SKT}[1]{#1}

\newcommand{\exampleend}{\begingroup\renewcommand{\qedsymbol}{\SKT{$\triangle$}}\qed\endgroup}

\definecolor{light-gray}{gray}{0.95}
\definecolor{mygreen}{rgb}{0,0.6,0}
\definecolor{mygray}{rgb}{0.5,0.5,0.5}
\definecolor{mymauve}{rgb}{0.58,0,0.82}
\definecolor{lightblue}{rgb}{0.8,0.9,1.0}
\definecolor{darkblue}{rgb}{0.6,0.8,1.0}
\definecolor{myyellow}{rgb}{1.0,0.95,0.6}

\newcommand{\OO}{\mathcal{O}}
\newcommand{\Aa}{\mathcal{A}}

\providecommand{\comment}[1]{}
\newcommand{\slim}[1]{{\scalebox{1.10}{$\scriptstyle #1$}}}

\tikzset{AUT style/.style={>=angle 60,every edge/.append style={},every state/.style={minimum size=16,inner sep=0}}}

\DeclareMathOperator{\im}{im}
\DeclareMathOperator{\row}{row}
\DeclareMathOperator{\supp}{supp}

\newcommand{\rk}{\operatorname{rk}}

\newcommand{\COL}{\mathsf{COL}}
\newcommand{\GL}{\mathrm{GL}}
\newcommand{\cL}{\mathcal{L}}
\newcommand{\M}{\mathcal{M}}
\newcommand{\N}{\mathbb{N}}
\newcommand{\NE}{\mathsf{NE}}
\newcommand{\Q}{\mathbb{Q}}
\newcommand{\ROW}{\mathsf{ROW}}
\newcommand{\U}{\mathcal{U}}
\newcommand{\UNIT}{\mathsf{UNIT}}
\newcommand{\V}{\mathcal{V}}
\newcommand{\Z}{\mathbb{Z}}

\def\WA{2.6} \def\WB{1.5} \def\WC{1.8} \def\WD{1.1} \def\WE{1.0}
\def\HA{2.6} \def\HB{1.8} \def\HC{2.0} \def\HD{1.6}

\pgfmathsetmacro{\Xzero}{0}
\pgfmathsetmacro{\Xone}{\WA}
\pgfmathsetmacro{\Xtwo}{\WA+\WB}
\pgfmathsetmacro{\Xthree}{\WA+\WB+\WC}
\pgfmathsetmacro{\Xfour}{\WA+\WB+\WC+\WD}
\pgfmathsetmacro{\Xfive}{\WA+\WB+\WC+\WD+\WE}
\pgfmathsetmacro{\Ytop}{0}
\pgfmathsetmacro{\Yone}{-\HA}
\pgfmathsetmacro{\Ytwo}{-(\HA+\HB)}
\pgfmathsetmacro{\Ythree}{-(\HA+\HB+\HC)}
\pgfmathsetmacro{\Yfour}{-(\HA+\HB+\HC+\HD)}

\title[The asymptotic size of finite irreducible semigroups of rational matrices]{The asymptotic size of finite irreducible semigroups of rational matrices}

\author[S.\ Kiefer]{Stefan Kiefer\lmcsorcid{0000-0003-4173-6877}}[a]
\address{Department of Computer Science, University of Oxford, UK}
\email{stefan.kiefer@cs.ox.ac.uk}

\author[A.\ Ryzhikov]{Andrew Ryzhikov\lmcsorcid{0000-0002-2031-2488}}[b]
\address{University of Warsaw, Poland}
\email{ryzhikov.andrew@gmail.com}
\thanks{Andrew Ryzhikov is supported by Polish National Science Centre SONATA BIS-12 grant number 2022/46/E/ST6/00230.}

\subjclass[2012]{Theory of computation~Formal languages and automata theory; Computing methodologies~Symbolic and algebraic manipulation}

\begin{document}

\begin{abstract}
In this paper we investigate the maximum size of finite semigroups of rational $n \times n$ matrices, with the goal of shedding more light on their structure. 
Such semigroups provide a rich generalisation of transition monoids of unambiguous (and, in particular, deterministic) finite automata.
While in general such semigroups can be arbitrarily large in terms of~$n$, a classical result of Sch{\"{u}}tzenberger from 1962 implies an upper bound of $2^{\OO(n^2 \log n)}$ for irreducible semigroups. A semigroup of rational matrices is called irreducible if the only subspaces of~$\Q^n$ that are invariant for all matrices in the semigroup are $\Q^n$ and the subspace consisting only of the zero vector.
Irreducible matrix semigroups can be viewed as the building blocks of general matrix semigroups, and as such play an important role in mathematics and computer science. From the point of view of automata theory, they can be seen as a generalisation of strongly connected weighted automata.
 
Using a very different technique from that of Sch{\"{u}}tzenberger, we improve the upper bound on the cardinality to~$3^{n^2}$.
This is the main result of the paper.
The bound is in some sense tight, as we show that there exists, for every $n$, a finite irreducible semigroup with $3^{\lfloor n^2/4 \rfloor}$ rational matrices.
Our main result also leads to an improvement of a bound, due to Almeida and Steinberg, on the mortality threshold of finite semigroups of rational matrices.
The mortality threshold is a number~$\ell$ such that if the zero matrix is in the semigroup, then the zero matrix can be written as a product of at most $\ell$ matrices from any subset that generates the semigroup.
\end{abstract}

\maketitle

\section{Introduction}\label{sec:intro-new}
Given a finite set~$\Aa$ of $n \times n$ matrices, the semigroup generated by $\Aa$ is the set of all products of matrices from~$\Aa$. Matrix semigroups appear naturally in many areas of computer science, such as automata theory, dynamical systems, program analysis and formal verification. Let us illustrate that with the following two examples.

\paragraph*{Two examples.}

Consider the linear loop in \Cref{fig:motivating} (left), where the conditions for both exiting the loop and choosing one of the two conditional branches are abstracted out and denoted by $*$. We assume that, in each iteration, one of the two linear operators is nondeterministically applied to the vector~$\big(\begin{smallmatrix} x \\ y
\end{smallmatrix}\big)$ of variables.
The overall set of linear operators that can be applied to this vector during the execution of the loop is thus the matrix semigroup generated by~$\Aa_\ell$ in \Cref{fig:motivating} (centre). This semigroup can be seen as the set of all behaviours of the loop.

\begin{figure}[ht]\centering
\begin{subfigure}[c]{0.35\textwidth} \centering
\fbox{\begin{minipage}{13.2em}
\texttt{\textbf{while} (*)}

\texttt{\ \ \textbf{if} (*)}

\texttt{\ \ \ \ (x, y) := (-y, x - y)}

\texttt{\ \ \textbf{else}}

\texttt{\ \ \ \ (x, y) := (x - y, -y)}
\end{minipage}}
\end{subfigure}
\hspace{5pt}
\begin{subfigure}[c]{0.30\textwidth} \centering
\begin{align*} \Aa_\ell = & 
\left\{\begin{pmatrix}
     0 & -1 \\ 1 & -1
\end{pmatrix}, \begin{pmatrix}
      1 & -1 \\ 0 & -1
\end{pmatrix}\right\}  \\ 
\Aa_r = & \left\{\begin{pmatrix}
     0 & 1 & 0 \\ 0 & 1 & 1 \\ 0 & 0 & 0
\end{pmatrix}, \begin{pmatrix}
      0 & 0 & 0 \\ 0 & 0 & 0 \\ 1 & 0 & 1
\end{pmatrix}\right\} 
\end{align*}
\end{subfigure}
\begin{subfigure}[c]{0.23\textwidth} \centering
\begin{tikzpicture} [node distance = 2cm]
\tikzset{every state/.style={inner sep=1pt,minimum size=1.5em}}

\node [state] at (0, 0) (1) {$1$};
\node [state] at (1, 1.5) (2) {$2$};
\node [state] at (2, 0) (3) {$3$};

\path [-stealth, thick]
(1) edge [bend left=10] node[left] {$a$} (2)
(2) edge [loop above] node[above] {$a$} (2)
(2) edge [bend left=10] node[right] {$a$} (3)

(3) edge [loop right] node[right] {$b$} (3)
(3) edge [bend left=10] node[below] {$b$} (1)
;
\end{tikzpicture}
\end{subfigure}
\caption{Left: a linear loop with nondeterministic branching. Right: an NFA. Centre: generating sets of the corresponding matrix semigroups.}\label{fig:motivating}
\end{figure}
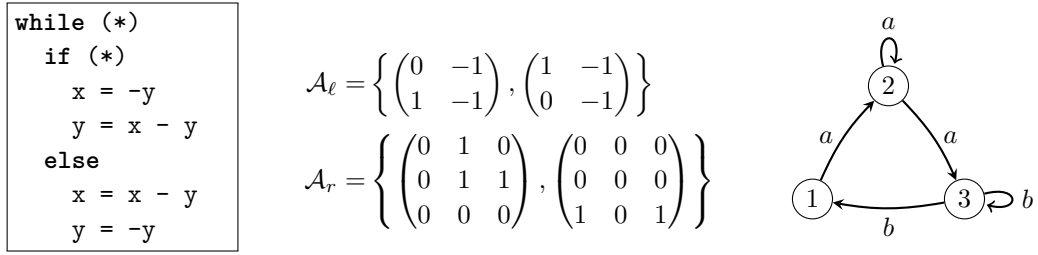

Similarly, the set of all possible behaviours of a nondeterministic finite automaton (NFA) is represented by its transition monoid. Note that the transition monoid does not depend on initial and final states. In this context, an NFA is called unambiguous if for every pair $p, q$ of states, each word labels at most one path from $p$ to $q$ in it. Clearly, every deterministic NFA is unambiguous. The transition monoid of an unambiguous NFA can be seen as the monoid generated by the set of transition matrices of the letters ($\Aa_r$ in the example in \Cref{fig:motivating} (centre)) \emph{with the usual addition and multiplication of the integers}. This fact allows us to consider unambiguous NFAs as automata with multiplicities (or, more generally, weighted automata), which significantly extends the variety of applicable techniques~\cite{Kiefer2025}.


\paragraph*{General motivation.}

In algorithmic applications, there is always a trade-off between the expressiveness of a model and the tractability of deciding its properties. This is especially important for matrix semigroups: for example, the question whether the semigroup generated by a given set of matrices contains the zero matrix is undecidable already for $3 \times 3$ integer matrices~\cite{Paterson1970}. For such problems, the known decidable special cases are usually obtained by restricting the dimension~\cite{Cassaigne2014,Colcombet2019,Bell2021} or considering only matrices with nonnegative entries~\cite{Ryzhikov2024RP}.

In this paper, we consider a different restriction: finiteness of the generated semigroup. It constitutes a ``middle ground'' between the two applications above. From the loop analysis point of view, it describes loops that have a finite set of behaviours regardless of the initial values of the variables. On the other hand, finite rational matrix semigroups are precisely the transition semigroups of minimal weighted automata over~$\Q$ with a finite image set~\cite{Mandel1977},
a rich generalisation of unambiguous and deterministic NFAs. Finite matrix semigroups are also building blocks of noncommutative power series of polynomial growth \cite[Section~9.2]{Berstel2011}. We remark that both $\Aa_\ell$ and $\Aa_r$ from \cref{fig:motivating} generate finite semigroups.


\paragraph*{The setting of the paper.} The general motivation behind the paper is to shed light on the structure of finite semigroups of $n \times n$ rational matrices and to develop new tools for analysing them. 
The concrete question we pursue here is about the maximum size of such semigroups in terms of~$n$. In general, there is no upper bound because for each $m \in \N$ the set $S_m \coloneqq \left\{ \big(\begin{smallmatrix} 0 & i \\ 0 & 0 \end{smallmatrix}\big) \mid 0 \le i < m \right\}$ forms a semigroup of size~$m$. However, intuitively, this example is in a sense degenerate, since only the same one-dimensional subspace is affected by the corresponding linear operators. Matrix semigroups where such degenerate behaviour does not occur are called irreducible (see the next section for the formal definition). They are actively studied in representation theory of finite monoids~\cite[Chapter 5]{Steinberg2016} and can be viewed as a generalisation of the concept of strongly connected finite automata to the case of weighted automata. Indeed, it is easy to see that a matrix semigroup is irreducible if and only if the weighted automaton corresponding to this matrix semigroup under every change of basis defined by a rational matrix is strongly connected.
In some applications, such as matrix mortality, a matrix semigroup can be directly analysed by decomposing it into irreducible semigroups of smaller dimension, see e.g.~\cite[Section~5]{AlmeidaSteinberg09} and the proof of \Cref{thm:min-rank-diameter} in the appendix.

\section{Existing results and our contributions}\label{sec:existing}

As usual, we denote by $\N$, $\Z$ and $\Q$ the sets of natural, integer and rational numbers, respectively.
We write $\GL_n(\Q)$ and $\GL_n(\Z)$ for the multiplicative group of all invertible $n \times n$ matrices over $\Q$ and~$\Z$, respectively. We denote by $\vec{v}$ a column vector of appropriate dimension, by $\vec{0}$ a zero column vector, by $A^\top$ the transpose of a matrix $A$, and by $O_n$ and~$I_n$ the $n\times n$ zero and identity  matrix, respectively. We assume that all vector spaces are over the field~$\Q$ and in particular that all matrices have rational entries, unless explicitly stated otherwise.

\paragraph*{The maximal size of finite matrix groups.}

Let us first highlight the importance of rational entries in our setting. Indeed, every cyclic group is isomorphic to a group generated by a $2 \times 2$ real rotation matrix, so there is no hope of bounding the size of finite matrix groups with real entries. The case of rational entries is however very different, and the maximal size of rational finite matrix groups is well understood.
By a folklore result (see, e.g., \cite[Theorem~1.6]{KuzmanovichPavlichenkov02}), any finite subgroup of~$\GL_n(\Q)$ is conjugate to a finite subgroup of~$\GL_n(\Z)$.
An elementary proof shows that the size of any finite subgroup of~$\GL_n(\Z)$ divides~$(2 n)!$; see, e.g., \cite[Chapter~IX]{Newman1972}.
Thus, denoting the size of the largest finite subgroup of~$\GL_n(\Q)$ by~$g(n)$, we have $g(n) \le (2 n)!$.
It is shown in a paper by Friedland~\cite{Friedland97} that $g(n) = 2^n n!$ holds for all sufficiently large~$n$.
This bound is attained by the group of signed permutation matrices (that is, matrices with entries in $\{-1, 0, 1\}$ with exactly one nonzero entry in each row and each column).
Friedland's proof rests on an article by Weisfeiler~\cite{Weisfeiler84} which in turn is based on the classification of finite simple groups.
Feit showed in an unpublished manuscript~\cite{FeitUnp} that $g(n) = 2^n n!$ holds if and only if $n \in \N \setminus \{2,4,6,7,8,9,10\}$; see also~\cite[Table~1]{Berry04} for a list of the maximal-size finite subgroups of~$\GL_n(\Q)$ for $n \in \{2,4,6,7,8,9,10\}$.
Feit's proof relies on an unpublished manuscript~\cite{WeisfeilerUnp} (also based on the classification of finite simple groups), which Weisfeiler left behind before his tragic disappearance.

\paragraph*{The maximal size of finite rational matrix semigroups.}

In view of the set $S_m$ from the previous section, bounds on the size of finite rational matrix semigroups either need to involve the number of generators (see, e.g., \cite{BumpusHKST20,Rida2025}) or an irreducibility assumption. A semigroup $S \subseteq \Q^{n \times n}$ is called \emph{irreducible} if the only vector spaces $\V \subseteq \Q^n$ such that $X \V \subseteq \V$ for all $X \in S$ are $\V = \Q^n$ and $\V = \{\vec{0}\}$. 
The semigroup~$S_m$ from above is not irreducible because for the vector space $\V = \{\big(\begin{smallmatrix} x \\ 0 \end{smallmatrix}\big) \mid x \in \Q\}$ we have $X \V = \{\big(\begin{smallmatrix} 0 \\ 0 \end{smallmatrix}\big)\} \subseteq \V$ for all $X \in S_m$.

Let us mention that the notion of irreducible matrices from nonnegative matrix theory, as in, e.g.,~\cite{Minc1988}, is weaker. Following~\cite{Minc1988}, a square matrix with nonnegative entries is called irreducible if permuting its rows and columns cannot result in a matrix of the shape $\big(\begin{smallmatrix}A & B \\ 0 & C\end{smallmatrix}\big)$, where~$A$ and~$C$ are square matrices. In terms of digraphs, a digraph is strongly connected if and only if its adjacency matrix is irreducible in this sense.
Clearly, a matrix generating an irreducible semigroup in our sense must be irreducible in the sense of~\cite{Minc1988}, but the converse is not always true.

In the book by Berstel and Reutenauer~\cite[Lemma~IX.1.2]{Berstel2011} it is shown that if a semigroup $S \subseteq \Q^{n \times n}$ is finite and irreducible then $|S| \le (2 n + 1)^{n^2} \in 2^{\OO(n^2 \log n)}$.
The technique, due to Sch{\"{u}}tzenberger~\cite{Schutzenberger62}, is based on the analysis of the coefficients of characteristic polynomials of the matrices in $S$, and in particular of their traces. In fact, the quantity $2 n + 1$ in the bound $(2 n + 1)^{n^2}$ corresponds to the possible traces in the set $\{-n, -n+1, \ldots, n\}$.

\paragraph*{The diameter and the mortality threshold.}

Let $S \subseteq \Q^{n \times n}$ be a finite semigroup generated by a subset $S_0 \subseteq S$. The length of a shortest product of elements from $S_0$ resulting in $X \in S$ is called the \emph{depth} of $X$. The \emph{diameter} of~$S$ is the maximum depth among all $X \in S$. Both the depth and the diameter are implicitly defined with respect to the set~$S_0$ of generators. Intuitively, the diameter indicates how fast one can reach any matrix. It is easy to see that the diameter of a finite semigroup cannot exceed its size.

\SKT{Bumpus et al.~\cite{BumpusHKST20} showed}, without assuming that $S$~is irreducible, that the diameter of $S$ with respect to any generating set is at most $2^{n (2 n + 3)} g(n)^{n+1} \in 2^{\OO(n^2 \log n)}$, where $g$~is the above-mentioned group-bound function with $g(n) \le (2 n)!$.
The technique used in~\cite{BumpusHKST20} is not based on traces but on exterior algebra. 
Nevertheless, their bound of $2^{\OO(n^2 \log n)}$ is strikingly similar to the aforementioned bound on semigroup cardinality. Panteleev~\cite{Panteleev15} showed that for every $n$ there exists a semigroup of diameter $2^{n + \Theta(\sqrt{n \log n})}$ with respect to some generating set. This semigroup is actually constructed as the transition monoid of a deterministic finite automaton, and thus consists of matrices with entries in $\{0, 1\}$ and exactly one nonzero entry in every row. As far as the authors know, no better lower bound is known for the maximum diameter of finite rational matrix semigroups.

The depth of the zero matrix (again, with respect to a set $S_0$ of generators) is called \emph{the mortality threshold} of $S$.
Intuitively, it indicates how fast one can reach the zero matrix.
Using a variation of the aforementioned technique due to Sch{\"{u}}tzenberger \cite{Schutzenberger62,Berstel2011}, Almeida and Steinberg~\cite{AlmeidaSteinberg09} showed that the mortality threshold of any finite rational matrix semigroup containing the zero matrix is at most $(2 n - 1)^{n^2}$ for $n \ge 2$.
This bound is once again $2^{\OO(n^2 \log n)}$, as in \cite{Schutzenberger62,Berstel2011,BumpusHKST20}.
The best known lower bound of~$\Omega(n^2)$ 
on the mortality threshold of finite semigroups of rational matrices is due to Rystsov~\cite{Rystsov1997}.
He conjectured that $\OO(n^2)$ is also the upper bound~\cite{Rystsov1992Rank}, which to the best of our knowledge has not been disproved. It is noteworthy that the lower bound again comes from the transition monoid of a DFA. 

\paragraph*{Our contributions and organization of the paper.}

Our main result, \cref{thm:main-bound}, is that any finite irreducible semigroup $S \subseteq \Q^{n \times n}$ has at most~$3^{n^2}$ elements, thus ``breaking'' the $2^{\OO(n^2 \log n)}$ barrier in previous results about both the cardinality and the diameter \cite{Schutzenberger62,Berstel2011,AlmeidaSteinberg09,BumpusHKST20}.
This is in a sense tight: as we show in~\cref{prop:lower}, any such bound needs to be at least~$3^{\lfloor n^2/4 \rfloor}$.
Recall that $|S|$~cannot be bounded purely in~$n$ without assuming irreducibility.
To showcase our technique, early on we give a relatively direct proof of the fact that if $S$~is also aperiodic (i.e., every subgroup of~$S$ has only one element), then $|S| \le 2^{n^2}$ (\Cref{thm:aper-bound}).
This result follows already from the proof of \cite[Theorem~5.8]{AlmeidaSteinberg09}, which used a different technique.
We also provide a lower bound of $2^{\lfloor n^2/4 \rfloor}$ for the aperiodic case (\Cref{prop:lower-nonnegative}). 
There are also finite irreducible matrix semigroups (over $\{-1,0,+1\}$) that do not contain the zero matrix but still have $2^{\Omega(n^2)}$ elements (\Cref{prop:lower-zero-free}).\footnote{This result was not contained in the conference version~\cite{KieferRyzhikovSTACS26} of this manuscript.}
Finally, as an application we show that if a finite, not necessarily irreducible, matrix semigroup contains the zero matrix, then its mortality threshold is at most~$3^{n^2}$ (\Cref{thm:min-rank-diameter}).
This improves the result by Almeida and Steinberg~\cite{AlmeidaSteinberg09}.

After the conference version~\cite{KieferRyzhikovSTACS26} of this manuscript was accepted, Steinberg~\cite{Steinberg26} provided a shorter proof of the upper bound~$3^{n^2}$, based on representation theory of finite semigroups. While his proof is more concise, we believe that our explicit linear-algebraic approach might have its own advantages. In particular, it might help with establishing stronger upper bounds on the values of the entries of matrices in irreducible finite semigroups, providing an improvement for the irreducible case of such a bound from~\cite{Rida2025}.
\SKT{We compare the two proofs in more detail at the end of \cref{sec:upper-two}.}

The paper is structured as follows. In \Cref{sec:irreducible} we establish basic facts about finite irreducible matrix semigroups and their (0-)minimal ideals. In \Cref{sub:G} we explain the construction of a group ``at'' (i.e., corresponding to) an idempotent from the (0-)minimal ideal. The results of \Cref{sec:irreducible,sub:G} are mostly known. 
\Cref {sec:upper-one,sec:upper-two} are dedicated to the proof of our main result (which we outline in the next paragraph below). Its application to matrix mortality can be found in \Cref{sec:mortality}. In \Cref{sec:lower} we provide lower bounds. We conclude the paper by highlighting some open problems in~\Cref{sec:conclusions}. Some proofs have been moved to an appendix for better readability of the manuscript.



\paragraph*{Technique.}
In contrast to~\cite{Schutzenberger62,AlmeidaSteinberg09,BumpusHKST20}, our technique is based neither on traces nor on exterior algebra. In fact, although the overall proof is non-trivial, it does not use anything outside of basic (semi)group theory and linear algebra. 
We outline our approach in the following.

Let $S$~be a finite rational $n \times n$ matrix semigroup, and let $T$~be a (0-)minimal ideal of~$S$.
One can show that all matrices in $T \setminus \{O_n\}$ have the same rank, say~$r$, which is the minimum nonzero rank in~$S$.
Given an idempotent $E \in T \setminus \{O_n\}$, fundamental semigroup theory (see \cref{sec:background-semigroup} for the background we need) describes a finite subgroup~$G$ of~$\GL_r(\Q)$ (often, e.g., in~\cite{Almeida2009representation}, called the maximal subgroup at~$E$), which reflects the symmetries in~$T$ (\Cref{sec:irreducible,sub:G}). As discussed above, the asymptotic size of such matrix groups is well understood.

We then construct an injective map $\Psi$ from $S$ to tuples of elements of $G \cup \{O_r\}$ (\cref{sub:Psi}). Thus, we have $|S| = |\Psi(S)|$ and so it suffices to bound the number of distinct tuples over $G \cup \{O_r\}$. This immediately leads to $|S| \le 3^{r^2 n^2}$, a bound that does not improve on~$2^{\OO(n^2 \log n)}$ unless the minimum rank $r$ is constant. However, in this way we already obtain a near-optimal bound on the cardinality of aperiodic semigroups (\Cref{sub:aperiodic}).

To strengthen the bound in the general case, we then show that the tuple elements are in a sense ``coupled'' via small matrix groups. 
This part (\cref{sub:width,sub:Hb,sub:row-prefix}) is the technical core of the paper and the most delicate aspect of the proof, even though it uses only elementary linear algebra.
Finally, the overall bound of $3^{n^2}$ is obtained by carefully counting the coupled tuples~$\Psi(X)$ \SKT{for $X \in S$} within a two-dimensional grid (\cref{sub:overall}).



\section{Irreducible matrix semigroups}\label{sec:irreducible}
Let $n \in \N$ and let $S \subseteq \Q^{n \times n}$ be a semigroup.
A vector space $\V \subseteq \Q^n$ is called \emph{$S$-invariant} if $S \V \subseteq \V$, i.e., $X \V \subseteq \V$ for all $X \in S$.
The semigroup~$S$ is called \emph{irreducible} if the only $S$-invariant subspaces of~$\Q^n$ are $\Q^n$ and $\{\vec{0}\}$.

\SKT{For the remainder of this section and in \cref{sub:G,sec:upper-one,sec:upper-two}, we fix~$S$ as a finite irreducible semigroup with $S \ne \{O_n\}$.}
The definition of irreducibility means that there are only trivial $S$-invariant ``column'' subspaces.
But it implies that there are also only trivial $S$-invariant ``row'' subspaces:
\begin{proposition} \label{prop:left-right}
Let $\U \subseteq \Q^{1 \times n}$ be a (row) vector space such that $\U S \subseteq \U$.
Then $\U = \Q^{1 \times n}$ or $\U = \{\vec{0}^\top\}$.
\end{proposition}
\begin{proof}
Suppose that $\U \neq \Q^{1 \times n}$, i.e., $\U$~is a proper subspace.
Define $\U^\circ \coloneqq \{\vec{v} \in \Q^n \mid \U \vec{v} = \{\vec{0}\}\}$.
Then $\dim \U + \dim \U^\circ = n$.
Since $\dim \U < n$, we have $\dim \U^\circ > 0$, i.e., $\U^\circ \ne \{\vec{0}\}$.
Let $X \in S$ and let $\vec{v} \in \U^\circ$.
Since $\U X \subseteq \U$, we have $\U X \vec{v} \subseteq \U \vec{v} = \{\vec{0}\}$, as $\vec{v} \in \U^\circ$.
Thus, $X \vec{v} \in \U^\circ$.
Since $\vec{v} \in \U^\circ$ was arbitrary, it follows that $X \U^\circ \subseteq \U^\circ$.
Since $X \in S$ was arbitrary, $\U^\circ$~is $S$-invariant.
Since $\U^\circ \ne \{\vec{0}\}$ and $S$~is irreducible, we have $\U^\circ = \Q^n$.
Since $\dim \U + \dim \U^\circ = n$, it follows that $\U = \{\vec{0}^\top\}$.
\end{proof}

We write $S^1 \coloneqq S \cup \{I_n\}$.

\subsection{The \mbox{(0-)}minimal ideal \texorpdfstring{\SKT{$T$}}{T}}\label{sub:T}

\SKT{By an \emph{ideal} of~$S$ we mean a nonempty subset $T \subseteq S$ satisfying $S T \subseteq T$ and $T S \subseteq T$ (a two-sided ideal).}
A \emph{minimal} ideal of a semigroup is an ideal that is minimal within the set of all ideals.
A \emph{0-minimal} ideal of a semigroup with zero is an ideal that is minimal within the set of all nonzero ideals.
Every finite semigroup has a minimal ideal, and 
every non-trivial finite semigroup with zero also has a 0-minimal ideal.
\SKT{If $O_n \notin S$, choose a minimal ideal~$T$ of~$S$; otherwise, choose a 0-minimal ideal~$T$ of~$S$.
In either case, $T \ne \{O_n\}$.
This choice of~$T$ remains fixed throughout the remainder of this section and \cref{sub:G,sec:upper-one,sec:upper-two}.}
We show the following two  lemmas.
\begin{lemma} \label{lem:T2}
We have $T^2 \ne \{O_n\}$.
\end{lemma}
\begin{proof}
%
%
Let $Z \in T \setminus \{O_n\}$ and choose $\vec{v} \in \Q^n$ such that $Z \vec{v} \ne \vec{0}$.
Let $\V \subseteq \Q^n$ be the vector space spanned by all $Y Z \vec{v}$, where $Y \in S^1$, i.e.,
\[
 \V \ \coloneqq \ \left\{\sum_{ Y \in S^1} \lambda_Y Y Z \vec{v} \ \middle\vert\ \text{all }\lambda_Y \in \Q \right\}.
\]
Since $I_n Z \vec{v} = Z \vec{v} \ne \vec{0}$, we have $\V \ne \{\vec{0}\}$.
To show that $\V$~is $S$-invariant, consider an arbitrary spanning vector of~$\V$, say $Y Z \vec{v} \in \V$ with $Y \in S^1$.
Then, for all $X \in S$ we have $X Y \in S$ and hence $X Y Z \vec{v} \in \V$.
Thus, by linearity, $\V$~is $S$-invariant.
Since $S$~is irreducible, it follows that $\V = \Q^n$.

Since $\V = \Q^n$ and $Z \ne 0$, there is $\vec{w} \in \V \setminus \ker Z$.
Write $\vec{w} = \sum_{Y \in S^1} \lambda_Y Y Z \vec{v}$ with all $\lambda_Y \in \Q$.
Since $\vec{w} \not\in \ker Z$, we have $Z \vec{w} \ne \vec{0}$.
Thus, there is $Y \in S^1$ with $\lambda_Y Z Y Z \vec{v} \ne \vec{0}$.
Hence, $Z Y Z \ne 0$.
Since $Z \in T$ and $T$~is an ideal, we have $Z Y \in T$.
It follows that $(Z Y) Z \in T^2 \setminus \{O_n\}$.
\end{proof}

\begin{lemma} \label{lem:same-rank}
All matrices in~$T \setminus \{O_n\}$ have the same rank $r \in \{1, \ldots, n\}$.
\end{lemma}
\begin{proof}
Pick $X \in T \setminus \{O_n\}$ of minimal nonzero rank.
Since $X \in T$ and $T$ is an ideal, we have $S^1 X S^1 \subseteq T$.
Moreover, $S^1 X S^1$ is an ideal of~$S$ and this ideal is nonzero, as it contains $X \ne 0$.
From the \mbox{(0-)}minimality of~$T$ we obtain $S^1 X S^1 = T$.
Hence, for any $Y \in T$ there exist $A, B \in S^1$ with $Y = A X B$.
For any $Y \in T \setminus \{O_n\}$,
$
\rk Y \ = \ \rk(A X B) \ \le \ \rk(X B) \ \le \ \rk X \ \le \ \rk Y\,,
$
using $\rk(C D)\le \min\{\rk C, \rk D\}$ and the minimality of~$\rk(X)$ among nonzero elements of $T$.
Hence all nonzero elements of $T$ have rank $\rk X$.
\end{proof}

\SKT{In \cref{sub:G,sec:upper-one,sec:upper-two}, we write~$r$ for this common rank, i.e., $\rk X = r$ for all $X \in T \setminus \{O_n\}$.}

\section{The idempotent \texorpdfstring{$E$}{E} and its group \texorpdfstring{$G$}{G}}
\label{sub:G}

\SKT{Recall that $S \subseteq \Q^{n \times n}$ is a fixed finite irreducible semigroup with $S \ne \{O_n\}$, that $T$ is its chosen \mbox{(0-)}minimal ideal, and that every matrix in $T \setminus \{O_n\}$ has rank~$r$.}
\SKT{The following lemmas use basic semigroup theory; the required background is given in \cref{sec:background-semigroup}.}

\begin{lemma} \label{lem:ideal-completely-0-simple}
\SKT{The ideal $T$ contains a nonzero idempotent.
For every idempotent $E \in T \setminus \{O_n\}$, the set $E T E \setminus \{O_n\}$ is a finite group with identity~$E$.}
\end{lemma}
\begin{proof}
It follows from \cref{prop:3.1.3,lem:T2} that the \mbox{(0-)}minimal ideal~$T$ is \mbox{(0-)}simple.
Since $T$~is finite, by \cref{prop:3.2.1}, it follows that $T$~is completely \mbox{(0-)}simple.
\SKT{In particular, $T$ contains a nonzero idempotent.
Let $E \in T \setminus \{O_n\}$ be any idempotent.}
By \cref{lem:group}, $E T E \setminus \{O_n\}$ is a group with identity~$E$.
The group is finite, as $E T E \setminus \{O_n\} \subseteq S$ and $S$~is finite.
\end{proof}

\begin{lemma} \label{lem:ESE=ETE}
\SKT{Let $E \in T \setminus \{O_n\}$ be an idempotent.
Then $E S E = E T E$ (which may contain $O_n$), and $E S E \setminus \{O_n\}$ is a finite group with identity~$E$.}
\end{lemma}
\begin{proof}
Since $E \in T$ and $T$~is an ideal, for all $X \in S$ we have $E X \in T$ and, hence, $E X E = E E X E \in E T E$.
Thus, $E S E \subseteq E T E$.
The converse inclusion is immediate from $T \subseteq S$.
\SKT{By \cref{lem:ideal-completely-0-simple}, $E T E \setminus \{O_n\}$ is a finite group with identity~$E$.}
\end{proof}

\SKT{Fix an idempotent $E \in T \setminus \{O_n\}$, which exists by \cref{lem:ideal-completely-0-simple}.}
The following lemma follows from the idempotence of~$E$.

\begin{lemma} \label{lem:CD}
There are matrices $D \in \Q^{n \times r}$ and $C \in \Q^{r \times n}$ with $E = D C$ and $C D = I_r$.
\end{lemma}
\begin{proof}[Proof of \Cref{lem:CD}]
Let $D \in \Q^{n \times r}$ be a matrix consisting of columns of~$E$ that form a basis of~$\im E$.
Since $E E = E$ and the columns of~$D$ are columns of~$E$, we have $E D = D$.
By the rank-nullity theorem, we have $\dim (\ker E) = n - r$.
Let $W \in \Q^{n \times (n-r)}$ be a matrix whose columns form a basis of~$\ker E$.
Thus, $E W = 0$.
Since $E E = E$, we have $\im E \cap \ker E = \{\vec{0}\}$, so the columns of the matrix $Q \in \Q^{n \times n}$ with $Q = \begin{pmatrix} D & W \end{pmatrix}$ are linearly independent.
Thus, $Q$~is invertible.
We have
\[
E Q \ = \ \begin{pmatrix} E D & E W \end{pmatrix} \ = \ \begin{pmatrix} D & 0 \end{pmatrix} \ = \ Q \begin{pmatrix} I_r & 0 \\ 0 & 0 \end{pmatrix}\,.
\]
Hence,
\[
E \ = \ Q \begin{pmatrix} I_r & 0 \\ 0 & 0 \end{pmatrix} Q^{-1}\,.
\]
Define $C \coloneqq \begin{pmatrix} I_r & 0 \end{pmatrix} Q^{-1}$ and recall that $D = Q \begin{pmatrix} I_r \\ 0 \end{pmatrix}$.
Then, as required, we have
\begin{align*}
C D \ &= \ \begin{pmatrix} I_r & 0 \end{pmatrix} Q^{-1} Q \begin{pmatrix} I_r \\ 0 \end{pmatrix} \ = \ I_r \quad \text{and} \\
D C \ &= \  Q \begin{pmatrix} I_r \\ 0 \end{pmatrix} \begin{pmatrix} I_r & 0 \end{pmatrix} Q^{-1} \ = \ Q \begin{pmatrix} I_r & 0 \\ 0 & 0 \end{pmatrix} Q^{-1} \ = \ E\,.\qedhere
\end{align*}
\end{proof}

The factorization $E = D C$ from \cref{lem:CD} allows us to put the group from \cref{lem:ESE=ETE} in a more succinct form, which will be useful when invoking bounds on the size of matrix groups.
To this end, fix $D \in \Q^{n \times r}$ and $C \in \Q^{r \times n}$ from \cref{lem:CD}, so that $D C = E$ and $C D = I_r$.
Define
\[
 G \ \coloneqq \ C S D \setminus \{O_r\} \ \subseteq \ \Q^{r \times r}\,.
\]
We have the following lemma.
\begin{lemma} \label{lem:group-G}
The set $G$ is a finite group, i.e., a finite subgroup of~$\GL_r(\Q)$.
Moreover, the finite group $E S E \setminus \{O_n\}$ from \cref{lem:ESE=ETE} is isomorphic to~$G$ via the isomorphism 
\[
\phi : E S E \setminus \{O_n\} \to G \quad \text{with} \quad \phi(X) \ \coloneqq \ C X D.
\]
\end{lemma}
\begin{proof}
Let us first consider a generalization of~$\phi$, namely the map $\Phi : E \Q^{n \times n} E \to \Q^{r \times r}$ with $\Phi(X) \coloneqq C X D$.
Note that $\Phi$~is a linear map and that
\[
 \Phi(E X E) \ = \ C D C X D C D \ = \ C X D \qquad \text{for all $X \in \Q^{n \times n}$.}
\]
The map~$\Phi$ has a trivial kernel, since if $C X D = \Phi(E X E) = O_r$ then $E X E = D C X D C = O_n$.
Thus, $\Phi$~and hence $\phi$ are injective.
It also follows that $\phi(E S E \setminus \{O_n\}) \subseteq G$.

Towards surjectivity of~$\phi$, let $X \in S$ with $C X D \ne 0$.
Then $\Phi(E X E) = C X D$.
If $E X E = O_n$ then $\Phi(E X E) = \Phi(O_n) = O_r \ne C X D$, a contradiction; hence $E X E \ne O_n$.
Thus, we also have $\phi(E X E) = \Phi(E X E) = C X D$.
It follows that $\phi(E S E \setminus \{O_n\}) = G$; i.e., $\phi$~is surjective.

It remains to show that $\phi$~is a homomorphism.
Using $C E = C D C = C$ and $E D = D C D = D$, we obtain $\phi(E)   =  C E D  =  C D C D  =  I_r$ and
\begin{align*}
\phi(E X E \cdot E Y E)  & \, = \,  C E X E Y E D \, = \,  C X E Y D \, = \, C X D C Y D \, = \, \phi(E X E) \cdot \phi(E Y E). \qedhere
\end{align*}
\end{proof}

\begin{example} \label{ex:1}
Set
\[
  C_1 \ \coloneqq \ \begin{pmatrix}1&0&0\\0&1&0\end{pmatrix},\qquad
  C_2 \ \coloneqq \ \begin{pmatrix}0&1&0\\0&0&1\end{pmatrix},\qquad
  D_1 \ \coloneqq \ \begin{pmatrix}1&0\\0&1\\1&0\end{pmatrix}, \qquad
  D_2 \ \coloneqq \ \begin{pmatrix}0&1\\1&0\\0&-1\end{pmatrix}\,.
\]
Then
\[
  C_1 D_1 \ = \ \begin{pmatrix}1&0\\0&1\end{pmatrix},\qquad
  C_1 D_2 \ = \ \begin{pmatrix}0&1\\1&0\end{pmatrix} \ = \  C_2 D_1,\qquad
  C_2 D_2 \ = \ \begin{pmatrix}1&0\\0&-1\end{pmatrix}\,.
\]
Let $G \subseteq \GL_2(\Q)$ be the group of signed $2 \times 2$ permutation matrices (order $8$).
Define 
\[
  S \ \coloneqq \ \{D_i g C_j \mid i,j \in \{1,2\},\ g \in G\}\,.
\]
The set~$S$ forms a semigroup, as for any $i,j,k,\ell \in \{1,2\}$ and $g,h \in G$,
\[
 (D_i g C_j) (D_k h C_\ell) \ = \ D_i (g C_j D_k h) C_\ell\,,
\]
and each $C_j D_k$ is in~$G$, as computed above.
The following facts about~$S$ can be checked:
(i)~$S$~has $8 \cdot 2 \cdot 2 = 32$ distinct elements, none of which is the zero matrix,
(ii)~all elements have rank $r=2$,
(iii)~$S$~is irreducible,
(iv)~$S$~is its own minimal ideal---equivalently, $S$~is simple,
(v)~$E \coloneqq D_1 C_1 \in S$ is an idempotent (recall that $C_1 D_1 = I_2$), and 
(vi)~$C_1 S D_1 = G$ (a group, as also implied by \cref{lem:group-G}).
\exampleend
\end{example}

\section{Upper bound: setting the stage}\label{sec:upper-one}
\subsection{The injective map \texorpdfstring{$\Psi$}{Psi}} \label{sub:Psi}

To explain our general approach, consider for the moment a map $\mu: S \to G \cup \{O_r\}$ with $\mu(X) = C X D$, a generalization of the map~$\phi$ from \cref{lem:group-G}.
We can use bounds on the group size mentioned in \Cref{sec:existing} to estimate~$|G|$.
But since $\mu$~is not in general injective, $|\mu(S)|$ does not bound~$|S|$.
Nevertheless, in the following we define a map~$\Psi$ with multiple components $\psi_{i j}$, each of which is a variant of~$\mu$.
More concretely, we have $\psi_{i j}: S \to G \cup \{O_r\}$ with $\psi_{i j}(X) = C U_i X V_j D$ for some matrices $U_i, V_j \in S^{1}$.
The matrices $U_i, V_j$ are chosen so that $\Psi : X \mapsto (\psi_{i j}(X))_{i j}$ \emph{is} injective.
Intuitively, the different $\psi_{i j}$ exhibit different ``group aspects'' of a semigroup element~$X$.
Since $\Psi$~is injective, we have $|S| = |\Psi(S)|$.
The known group bounds then help to estimate~$|\Psi(S)|$. We provide further intuition of our approach at the end of this subsection. 

Since $\sum_{X \in S^1} \im (X D)$ is $S$-invariant and nonzero (it contains $\im D$) and $S$~is irreducible, we have $\sum_{X \in S^1} \im (X D) = \Q^n$.
Thus, there exist $V_1, \ldots, V_v \in S^1$ ($v \ge 1$) such that
\[
 \im(V_1 D) + \cdots + \im(V_v D) \ = \ \Q^n\,.
\]
As we will see later, this sum need not be direct.

Dually (cf.\ \cref{prop:left-right}), there exist $U_1, \ldots, U_u \in S^1$ ($u \ge 1$) such that
\[
 \row(C U_1) + \cdots + \row(C U_u) \ = \ \Q^{1 \times n}\,. 
\]
For $0 \le a \le u$ and $0 \le b \le v$ define the vector spaces
\begin{align*}
 \U_{a} \ &\coloneqq \ \row(C U_1) + \cdots + \row(C U_{a}) \ \subseteq \ \Q^{1 \times n} \qquad \text{and} \\
 \V_{b} \ &\coloneqq \ \im(V_1 D) + \cdots + \im(V_{b} D) \ \subseteq \ \Q^n\,.
\end{align*}
By convention, $\U_0 = \{\vec{0}^\top\}$ and $\V_0 = \{\vec{0}\}$.
Without loss of generality, we can assume for all $1 \le a \le u$ that $\row(C U_a) \not\subseteq \U_{a-1}$.
Similarly, we also assume for all $1 \le b \le v$ that $\im(V_b D) \not\subseteq \V_{b-1}$. Thus, the vector space inclusions $\{\vec{0}\} = \U_0 \subset \U_1 \subset \ldots \subset \U_u = \Q^{1 \times n}$ are strict, and similarly for the~$\V_j$.
It follows that $u, v \le n$.
We note the following lemma.
\begin{lemma} \label{lem:rk-VjD}
For all $1 \le j \le v$ we have $\rk(V_j D) = r$.
\end{lemma}
\begin{proof}
Recall that $E = D C$.
Thus $\im E \subseteq \im D$.
Since $\rk E = r = \rk D$, we have $\im D = \im E$.
Hence, $\im (V_j D) = \im (V_j E)$, and so $\rk (V_j D) = \rk (V_j E)$.
Since $\im(V_j D) \not\subseteq \V_{j-1}$, we have $V_j D \ne O_n$.
Thus, $V_j E \ne O_n$.
Since $E \in T$ and $T$~is an ideal, we also have $V_j E \in T$.
Since $r$~is the common rank among nonzero elements of~$T$, it follows that $\rk (V_j E) = r$.
Hence, $\rk (V_j D) = \rk (V_j E) = r$.
\end{proof}

For $1 \le i \le u$ and $1 \le j \le v$ and $X \in S$ define
\[
 \psi_{i j}(X) \ \coloneqq \ C U_i X V_j D \ \in \ G \cup \{O_r\}\,.
\]
Also define
\[
 \Psi(X) \ \coloneqq \ \begin{pmatrix} \psi_{1 1}(X) & \cdots & \psi_{1 v}(X) \\
                                \vdots        & \ddots & \vdots\\
                                \psi_{u 1}(X) & \cdots & \psi_{u v}(X)
                \end{pmatrix} \ = \ 
                \begin{pmatrix} C U_1 X V_1 D & \cdots & C U_1 X V_v D \\
                                \vdots        & \ddots & \vdots\\
                                C U_u X V_1 D & \cdots & C U_u X V_v D
                \end{pmatrix} \ \in \ \Q^{u r \times v r}\,.
\]
We will primarily view~$\Psi(X)$ not as a large matrix, but as a grid (or array) of smaller $\Q^{r \times r}$ matrices.

\begin{lemma} \label{lem:Psi-injective}
The map~$\Psi$ is injective.
\end{lemma}
\begin{proof}
Let 
$
 V \ \coloneqq \ \begin{pmatrix} V_1 D & \cdots & V_v D \end{pmatrix} \ \in \ \Q^{n \times v r}\,.
$

It follows from the definition of $V_1, \ldots, V_v$ that $\im V = \sum_{j=1}^v \im(V_j D) = \Q^n$.
Thus, $\rk V = n$ and so $V$~has a right inverse $V' \in \Q^{v r \times n}$ with $V V' = I_n$.
Dually, the matrix
\[
 U \ \coloneqq \ \begin{pmatrix} C U_1 \\ \vdots \\ C U_u \end{pmatrix} \ \in \ \Q^{u r \times n}
\]
has a left inverse $U' \in \Q^{n \times u r}$ with $U' U = I_n$.
Noting that $\Psi(X) = U X V$, we have
\[
 U' \Psi(X) V' \ = \ U' U X V V' \ = \ I_n X I_n \ = \ X\,.
\]
It follows that $\Psi$~is injective.
\end{proof}

\begin{example} \label{ex:2}
We continue \cref{ex:1}.
Choose
\[ 
 U_1 \ \coloneqq \ I_3 \text{ (recall that $I_3 \in S^1$)}, \quad 
 U_2 \ \coloneqq \ D_1 C_2\,,\quad
 V_1 \ \coloneqq \ I_3\,,\quad
 V_2 \ \coloneqq \ D_2 C_1\,, \quad \text{so that}
\]
\[
 C_1 U_1  =  C_1  =  \left(\begin{smallmatrix}1&0&0\\0&1&0\end{smallmatrix}\right),\
 C_1 U_2  =  C_2  =  \left(\begin{smallmatrix}0&1&0\\0&0&1\end{smallmatrix}\right),\
 V_1 D_1  =  D_1  =  \left(\begin{smallmatrix}1&0\\0&1\\1&0\end{smallmatrix}\right),\
 V_2 D_1  =  D_2  =  \left(\begin{smallmatrix}0&1\\1&0\\0&-1\end{smallmatrix}\right)\,.
\]
Thus, $\row(C_1 U_1) + \row(C_1 U_2) = \Q^{1 \times 3}$ and $\im(V_1 D_1) + \im(V_2 D_1) = \Q^{3}$.
The ranks of these four matrices are all~$2$, consistent with \cref{lem:rk-VjD} and its analogue for $C_1 U_1, C_1 U_2$.
We have $u = v = 2$.
\exampleend
\end{example}

In the following we will be interested in~$|S|$.
By \cref{lem:Psi-injective}, we have $|S| = |\Psi(S)|$; i.e., it suffices to estimate the number of different $\Psi(X)$ for $X \in S$.
Taking this further, we have
\begin{align*}
|S| \ &= \ |\Psi(S)| \ \le \ \prod_{i=1}^u \prod_{j=1}^v |\psi_{i j}(S)| \ = \ \prod_{i=1}^u \prod_{j=1}^v |C U_i S V_j D| 
 \ \le \ \prod_{i=1}^u \prod_{j=1}^v (|G|+1) \ \\ 
&= \ (|G|+1)^{u v} \ \le \ 3^{r^2 u v} \qquad \text{using \cref{lem:group-G,lem:bound-on-group-size}.}
\end{align*}
Suppose for a moment that the vector spaces $(\row(C U_i))_i$ were independent and the vector spaces $(\im(V_j D))_j$ were independent,
 i.e., suppose that $\bigoplus_{i=1}^u \row(C U_i) = \Q^{1 \times n}$ and $\bigoplus_{j=1}^v \im(V_j D) = \Q^{n}$.
Then $u r = n$ and $v r = n$ (in particular, $r \mid n$ and $u = v = n/r$), and using the inequality above we would obtain $|S| \le 3^{n^2}$.
However, in general this independence does not hold and we have to estimate $u,v \le n$, giving only a weaker bound $|S| \le 3^{r^2 n^2}$.
Therefore, we pursue a different avenue, based on the idea that if, say, $\im(V_1 D)$ and $\im(V_2 D)$ overlap nontrivially then $\psi_{i 1}(X)$ and~$\psi_{i 2}(X)$ are ``coupled'' across $X \in S$, i.e., $\psi_{i 1}(X)$ and~$\psi_{i 2}(X)$ do not vary independently.
Formalizing this idea and making it work is the key technical contribution of this paper. Before doing that, we treat a much easier case of aperiodic semigroups.

\subsection{The aperiodic case} \label{sub:aperiodic}

A semigroup is called \emph{aperiodic} if every subsemigroup which is also a group is trivial, i.e., has only one element \SKT{(its identity)}.
In this subsection we analyze the size of~$S$, assuming it is aperiodic.
Showcasing the use of our injective map~$\Psi$, \cref{thm:aper-bound} below recovers a result from \cite{AlmeidaSteinberg09}.

\begin{lemma} \label{lem:aper-rank}
If $S$~is aperiodic, we have $r=1$ and $G = \{I_1\}$.
\end{lemma}
\begin{proof}
\SKT{Recall that $E \in T \setminus \{O_n\}$ is the idempotent fixed in \cref{sub:G}.
The matrices $C,D$ chosen there satisfy $E = D C$ and $C D = I_r$ (\cref{lem:CD}).}
By \cref{lem:ESE=ETE}, $E S E \setminus \{O_n\}$ is a group.
Thus, it is a subgroup of the semigroup~$S$.
Since $S$~is aperiodic, $E S E \setminus \{O_n\} = \{E\}$.
By \cref{lem:group-G} it follows that $G = \{I_r\}$.

Let $\vec{d} \in \Q^n$ be a (necessarily nonzero) column of~$D$.
Since $E D = D C D = D$, we have $E \vec{d} = \vec{d}$.
Write $\langle \cdot \rangle$ for the $\Q$-span.
The vector space $\langle S \vec{d} \rangle$ is $S$-invariant and nonzero, as it contains~$E \vec{d} = \vec{d}$.
Thus, irreducibility of~$S$ implies that $\langle S \vec{d} \rangle = \Q^n$.
Therefore,
\[
 \im E \ = \ E \Q^n \ = \ E \langle S \vec{d} \rangle \ \mathop{=}^{E \vec{d} = \vec{d}} \ \langle E S E \vec{d} \rangle
 \ \mathop{=}^{E S E \setminus \{O_n\} = \{E\}} \ \langle E \vec{d} \rangle \ \mathop{=}^{E \vec{d} = \vec{d}} \ \langle \vec{d} \rangle\,.
\]
Hence, $r = \rk E = \dim (\im E) = 1$, and so $G = \{I_1\}$.
\end{proof}

Using the injective map~$\Psi$ we obtain the following result, which essentially follows from the proof of~\cite[Theorem~5.8]{AlmeidaSteinberg09}. This is already a good illustration how our approach differs from the approach of \cite{Schutzenberger62,Berstel2011,AlmeidaSteinberg09}: instead of bounding the possible values that the traces of matrices in $S$ can take, we consider a family $\Psi$ of ``linear'' maps from $S$ to the group $G$.

\begin{theorem}\label{thm:aper-bound}
Let $S \subseteq \Q^{n \times n}$ be an aperiodic finite irreducible semigroup. Then $|S| \le 2^{n^2}$.
\end{theorem}
\begin{proof}
By \cref{lem:aper-rank}, we have $r=1$.
Using the assumption that the spaces $\U_{i}$ for $0 \le i \le u$ are strictly increasing, and similarly for the~$\V_{j}$, it follows that $u = n = v$.
Also by \cref{lem:aper-rank}, we have $\psi_{i j}(X) \in \{O_1,I_1\}$ for all $1 \le i,j \le n$ and all $X \in S$; i.e., $\Psi(S) \subseteq \{O_1,I_1\}^{n \times n}$.
Since $\Psi$~is injective, $|S| = |\Psi(S)| \le |\{O_1,I_1\}^{n \times n}| = 2^{n^2}$.
\end{proof}

\section{Upper bound: doing the main work}\label{sec:upper-two}

\SKT{To bound $|S|$, we count the distinct values of~$\Psi(X)$ as $X$ ranges over~$S$.
In \cref{sub:width,sub:Hb}, we study the overlaps between the spaces $\im(V_j D)$ and establish bounds that we will use in this count.
In \cref{sub:row-prefix}, we use these bounds to count the possible values of a block once the preceding blocks in its row are fixed.
Finally, in \cref{sub:overall}, we apply the same argument to columns and combine the row and column bounds to count the possible values of~$\Psi(X)$.}

\subsection{The width \texorpdfstring{$w_b$}{} of a column} \label{sub:width}

For this and the next two subsections, we fix an arbitrary column index $b \in \{1, \ldots, v\}$.
Define the \emph{width}~$w_b$ of block column~$b$:
\[
w_b \ \coloneqq \ \dim \V_b - \dim \V_{b-1} \ = \ \dim(\V_{b-1} + \im(V_b D)) - \dim \V_{b-1}\,.
\]
Intuitively, $\im(V_b D)$ adds $w_b \le r$ independent dimensions to~$\V_{b-1}$.
We note that
\begin{equation} \label{eq:wj}
\begin{aligned}
 w_1 \ &= \ \dim \V_1 \ = \ \rk(V_1 D) \ = \ r    && \text{by \cref{lem:rk-VjD}\quad and}\\ 
 w_1 + \cdots + w_v \ &= \ \dim \V_v \ = \ n      && \text{from the definition of the $\V_j$.}
\end{aligned}
\end{equation}
The following picture illustrates these widths.
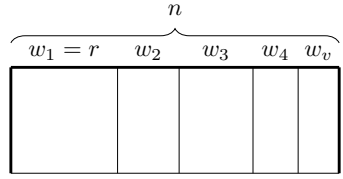
\begin{center}
\scalebox{.9}{
\begin{tikzpicture}[x=6mm,y=6mm]
  \draw[very thick] (\Xzero,\Ytop) -- (\Xfive,\Ytop);   
  \draw            (\Xzero,\Yone) -- (\Xfive,\Yone);    
  \draw[very thick] (\Xzero,\Ytop) -- (\Xzero,\Yone);   
  \draw[very thick] (\Xfive,\Ytop) -- (\Xfive,\Yone);   
  \foreach \XX in {\Xone,\Xtwo,\Xthree,\Xfour} \draw (\XX,\Ytop) -- (\XX,\Yone);

  \path ({0.5*(\Xzero+\Xone)},0.35)  node {$w_1=r$};
  \path ({0.5*(\Xone+\Xtwo)},0.35)   node {$w_2$};
  \path ({0.5*(\Xtwo+\Xthree)},0.35) node {$w_3$};
  \path ({0.5*(\Xthree+\Xfour)},0.35)node {$w_4$};
  \path ({0.5*(\Xfour+\Xfive)},0.35) node {$w_v$};

  \draw[decorate,decoration={brace,amplitude=6pt}]
    (\Xzero,0.6) -- node[above=8pt] {$n$} (\Xfive,0.6);
\end{tikzpicture}}
\end{center}
Also define
\[
 \cL_b \ \coloneqq \ \{\vec{y} \in \Q^r \mid V_b D \vec{y} \in \V_{b-1}\} \qquad \text{and} \qquad \ell_b \ \coloneqq \ \dim \cL_b\,.
\]
In words, $\cL_b$ is the vector space consisting of the vectors $\vec{y} \in \Q^{r}$ that the matrix~$V_b D$ maps into the intersection of $\V_{b-1}$ and $\im(V_b D)$; i.e., we have
$
 V_b D \cL_b = \V_{b-1} \cap \im (V_b D)
$.
The following lemma connects $w_b$ and~$\ell_b$.
\begin{lemma} \label{lem:lw}
We have $w_b = r -\ell_b > 0$.
\end{lemma}
\begin{proof}
Consider the map $V_b D : \Q^r \to \Q^n$.
Its domain is $r$-dimensional and we have $\rk (V_b D) = r$ by \cref{lem:rk-VjD}.
It follows that the map $V_b D$ is injective.
Hence its restriction to any subspace is injective.
Thus, $\dim (V_b D \cL_b) = \dim \cL_b = \ell_b$.
Hence, 
\begin{align*}
w_b \ &= \ \dim(\V_{b-1} + \im(V_b D)) - \dim \V_{b-1} \\ 
&= \ \dim \V_{b-1} + \rk(V_b D) - \dim (\V_{b-1} \cap \im (V_b D)) - \dim \V_{b-1} \\
&= \ r - \dim (\V_{b-1} \cap \im (V_b D)) && \text{by \cref{lem:rk-VjD}}\\
&= \ r - \dim (V_b D \cL_b) \ = \ r - \ell_b\,.
\end{align*}
Since $\cL_b \subseteq \Q^r$, we have $\ell_b \le r$.
If $\ell_b = r$ then $\cL_b = \Q^r$, implying that $\im (V_b D) = V_b D \cL_b \subseteq \V_{b-1}$, contradicting the assumption made after the definition of $\V_{b}$.
Hence, $\ell_b < r$.
\end{proof}

\begin{example} \label{ex:3}
Continuing \cref{ex:2}, we have
\[
\V_1 \ = \ \im(V_1 D_1) \ = \ \left\{\left(\begin{smallmatrix}p \\ q \\ p \end{smallmatrix}\right) \;\middle\vert\; p,q \in \Q \right\}\,, \qquad \im(V_2 D_1) \ = \ \left\{\left(\begin{smallmatrix}p \\ q \\ -p \end{smallmatrix}\right) \;\middle\vert\; p,q \in \Q \right\}\,.
\]
Thus, $\V_2 = \V_1 + \im(V_2 D_1) = \Q^3$.
Hence,
\[
 v = 2, \quad w_1 = \dim \V_1 = 2 = r, \quad w_2 = \dim \V_2 - \dim \V_1 = 3 - 2 = 1.
\]
We also have
\begin{align*}
\cL_2 \ 
&= \ \{\vec{y} \in \Q^2 \mid V_2 D_1 \vec{y} \in \V_1\}
 \ = \  \left\{\big(\begin{smallmatrix}p \\ q \end{smallmatrix}\big) \in \Q^2 \;\middle\vert\; \left(\begin{smallmatrix}0&1\\1&0\\0&-1\end{smallmatrix}\right) \big(\begin{smallmatrix}p \\ q \end{smallmatrix}\big) \in \V_1\right\} \\
&= \  \left\{\big(\begin{smallmatrix}p \\ q \end{smallmatrix}\big) \in \Q^2 \;\middle\vert\; \left(\begin{smallmatrix}q \\ p \\ -q\end{smallmatrix}\right)  \in \V_1\right\}
\ = \ \left\{\big(\begin{smallmatrix}p \\ 0 \end{smallmatrix}\big) \;\middle\vert\; p \in \Q\right\}\,.
\end{align*}
Thus, $\ell_2 = \dim \cL_2 = 1$, matching \cref{lem:lw}: $w_2 = r-\ell_2 =2-1 = 1$.
\exampleend
\end{example}

\subsection{The coupling group \texorpdfstring{$H_b$}{Hb}} \label{sub:Hb}

In this subsection we introduce $H_b$, a subgroup of~$G$.
\SKT{We first outline why this group will be useful; the claims in this outline will be proved in \cref{sub:row-prefix}.
Consider one block row~$a$ of~$\Psi(X)$, with $X$ ranging over~$S$.
We will show that fixing the earlier blocks $\psi_{a 1}(X), \ldots, \psi_{a (b-1)}(X)$ determines $\psi_{a b}(X)\vec{y}$ for every $\vec{y} \in \cL_b$ (\cref{lem:identity-on-Lb}).
The key point is that $V_b D \vec{y}$ belongs to~$\V_{b-1}$, which is spanned by the columns of $V_1 D, \ldots, V_{b-1} D$.
Thus, all possible values of the next block~$\psi_{a b}(X)$ will agree on~$\cL_b$.}

\SKT{If $\widehat{g},g \in G$ are two nonzero possible values of this next block, their agreement on~$\cL_b$ means that $\widehat{g}^{-1}g$ fixes every vector in~$\cL_b$.
This motivates the definition of the \emph{coupling group}}
\[
H_b \ \coloneqq \ \{g \in G \mid g \vec{y} = \vec{y} \text{ for all } \vec{y} \in \cL_b\}\,;
\]
\SKT{i.e., $H_b$ consists of those matrices in~$G$ that fix every vector in~$\cL_b$.}

\SKT{As we will prove in \cref{sub:row-prefix}, once the earlier blocks are fixed, every nonzero possible value of~$\psi_{a b}(X)$ has the form $\widehat{g}h$ with $h \in H_b$, for any one such value~$\widehat{g}$.
There are therefore at most $|H_b|$ nonzero possibilities, plus possibly the zero matrix.
In this subsection we will bound $|H_b|+1$ by~$3^{w_b^2}$.
The smaller the width~$w_b$, the smaller the bound on~$|H_b|$ becomes, and hence the smaller our bound on the number of possible values of~$\psi_{a b}(X)$.}

\begin{example} \label{ex:4}
Continuing \cref{ex:3}, we have
\[
H_2 
\ = \ \{g \in G \mid g \vec{y} = \vec{y} \ \; \forall\,\vec{y} \in \cL_2\}
\ = \ \{g \in G \mid g \big(\begin{smallmatrix}p \\ 0 \end{smallmatrix}\big) = \big(\begin{smallmatrix}p \\ 0 \end{smallmatrix}\big) \ \; \forall\, p \in \Q\}
\ = \ \left\{\left(\begin{smallmatrix}1&0\\0&1\end{smallmatrix}\right),\ \left(\begin{smallmatrix}1&0\\0&-1\end{smallmatrix}\right)\right\}\,.
\]
Thus, $|H_2|+1 = 3 = 3^{1^2} = 3^{w_2^2}$, realizing the upper bound of the following lemma.
\exampleend
\end{example}

Our analysis of the semigroup size is based on known bounds on the size of finite matrix groups.
\SKT{The following bound is classical.
For completeness, we give a proof in the appendix, following the exposition in~\cite{KuzmanovichPavlichenkov02}.}

\begin{lemma}\label{lem:bound-on-group-size}
Let $n \ge 1$.
Any finite subgroup of~$\GL_n(\Q)$ has at most $3^{n^2}-1$ elements.
\end{lemma}

\begin{remark} \label{rem:bound-on-group-size}
\Cref{lem:bound-on-group-size} is not tight.
As mentioned in the introduction, it is known (via an elementary proof not based on the classification of finite simple groups) that the order of any finite subgroup, say~$H$, of~$\GL_n(\Q)$ divides~$(2 n)!$ (see, e.g., \cite[Chapter~IX]{Newman1972}); so $|H| \le (2 n)! = 3^{\Theta(n \log n)}$.
It is not difficult to prove \cref{lem:bound-on-group-size} by showing that $(2 n)! \le 3^{n^2}-1$, but the more fundamental proof of \cref{lem:bound-on-group-size} in the appendix might give more insight.

We will use the group bound from \cref{lem:bound-on-group-size} to bound the size of finite irreducible matrix semigroups $S \subseteq \Q^{n \times n}$ in terms of~$n$.
An important role will be played by a certain group that is isomorphic to a finite subgroup of~$\GL_r(\Q)$, where $r$~is the minimum nonzero rank of the matrices in~$S$.
The bottleneck (for our semigroup bound in terms of~$n$) will turn out to be the case $r=1$, where we have $3^{1^2}-1 = 2 = (2 \cdot 1)!$.
Therefore, the mentioned asymptotically tighter results for groups do not improve our main result on semigroups.\qed
\end{remark}

\begin{lemma} \label{lem:Hb-bound}
We have
$
 |H_b| + 1 \ \le \ 3^{w_b^2}\,.
$
\end{lemma}

\begin{proof}
For any $h_1, h_2 \in H_b$ we have $h_1 h_2 \in H_b$, as $h_1 h_2 \vec{y} = h_1 \vec{y} = \vec{y}$ holds for all $\vec{y} \in \cL_b$.
Let $h \in H_b$.
We show that $h^{-1} \in H_b$.
Indeed, for all $\vec{y} \in \cL_b$ we have $h^{-1} \vec{y} = h^{-1} (h \vec{y}) = (h^{-1} h) \vec{y} = \vec{y}$.
We conclude that $H_b$~is a group.
It is finite, as $G \supseteq H_b$ is finite.

Let $\M_b \subseteq \Q^{r}$ be a vector space such that $\Q^r = \cL_b \oplus \M_b$ (so that $\dim \M_b = r - \ell_b = w_b > 0$ using \cref{lem:lw}).
In a basis adapted to $\Q^r = \cL_b \oplus \M_b$, every $h \in H_b$ has the block form
\[
 h \ = \ \begin{pmatrix} I_{\ell_b} & N \\ 0 & Q \end{pmatrix} \qquad \text{where } Q \in \GL_{w_b}(\Q),  \ N : \M_b \to \cL_b\,,
\]
as $h$~fixes every $\vec{y} \in \cL_b$.

For $h \in H_b$ write $\pi(h) \coloneqq Q \in \GL_{w_b}(\Q)$; \SKT{i.e., $\pi(h)$ is the lower-right block of~$h$ in the chosen basis.}
We show that $\pi : H_b \to \pi(H_b)$ is a group isomorphism from $H_b$ to~$\pi(H_b)$.
Trivially, $\pi$~is surjective onto~$\pi(H_b)$.
It is also a homomorphism, as
\[
\pi \left( \begin{pmatrix} I_{\ell_b} & N_1 \\ 0 & Q_1 \end{pmatrix} \begin{pmatrix} I_{\ell_b} & N_2 \\ 0 & Q_2 \end{pmatrix} \right) 
\ = \ Q_1 Q_2 \ = \ 
\pi \begin{pmatrix} I_{\ell_b} & N_1 \\ 0 & Q_1 \end{pmatrix} \pi \begin{pmatrix} I_{\ell_b} & N_2 \\ 0 & Q_2 \end{pmatrix}\,.
\]
Towards an application of \cref{lem:group-ker}, let $h \in H_b$ with $\pi(h) = I_{w_b}$, i.e., 
\[
 h \ = \ \begin{pmatrix} I_{\ell_b} & N \\ 0 & I_{w_b} \end{pmatrix} \qquad \text{for some matrix~$N$.}
\]
By an easy induction, it follows that
\[
 h^m \ = \ \begin{pmatrix} I_{\ell_b} & m N \\ 0 & I_{w_b} \end{pmatrix} \qquad \text{for all $m \ge 0$.}
\]
Since $G$~is a finite group, we have $h^m = I_r$ for some $m \ge 1$.
Thus, $m N = 0$.
Dividing by~$m$, we obtain $N = 0$; hence $h = I_r$.
Using \cref{lem:group-ker} it follows that $\pi : H_b \to \pi(H_b)$ is a group isomorphism and that $\pi(H_b)$~is a finite subgroup of~$\GL_{w_b}(\Q)$.

Thus, $H_b$ is isomorphic to a finite subgroup of~$\GL_{w_b}(\Q)$.
By \Cref{lem:bound-on-group-size}, $|H_b| + 1 \le 3^{w_b^2}$.
\end{proof}

\subsection{A block row prefix} \label{sub:row-prefix}

For this subsection, we fix an arbitrary block row index $a \in \{1, \ldots, u\}$.
We consider the (number of) possible first $b$~blocks of the $a$th block row of~$\Psi(X)$ when $X$~ranges over~$S$, i.e., the possible
\[
\begin{pmatrix} \psi_{a 1}(X) & \cdots & \psi_{a b}(X) \end{pmatrix} \qquad \text{where $X \in S$.}
\]
\SKT{The following lemma shows that fixing the earlier blocks $\psi_{a 1}(X), \ldots, \psi_{a (b-1)}(X)$ determines $\psi_{a b}(X)\vec{y}$ for every $\vec{y} \in \cL_b$.}

\begin{lemma} \label{lem:identity-on-Lb}
There exist linear maps $\Theta_1, \ldots, \Theta_{b-1} : \V_{b-1} \to \Q^r$ such that 
\[
 \psi_{a b}(X) \vec{y} \ = \ \sum_{j=1}^{b-1} \psi_{a j}(X) \Theta_j(V_b D \vec{y}) \qquad \text{for all $X \in S$ and all $\vec{y} \in \cL_b$.}
\]
\end{lemma}
\begin{proof}
Consider the linear map
\[
\Omega : (\Q^r)^{b-1} \to \V_{b-1} \qquad \text{with} \qquad \Omega(\vec{y}_1, \ldots, \vec{y}_{b-1}) \ \coloneqq \ \sum_{j=1}^{b-1} V_j D \vec{y}_j\,.
\]
Since $\Omega$~is surjective, it has a linear right inverse $\Sigma : \V_{b-1} \to (\Q^r)^{b-1}$ with $\vec{z} = \Omega(\Sigma(\vec{z}))$ for all $\vec{z} \in \V_{b-1}$.
Write $\Sigma(\vec{z}) \eqqcolon (\Theta_1(\vec{z}), \ldots, \Theta_{b-1}(\vec{z}))$.
Thus, $\vec{z} = \sum_{j=1}^{b-1} V_j D \Theta_j(\vec{z})$ for all $\vec{z} \in \V_{b-1}$.
In particular, for all $\vec{y} \in \cL_b$, since $V_b D \vec{y} \in \V_{b-1}$,
\[
 V_b D \vec{y} \ = \ \sum_{j=1}^{b-1} V_j D \Theta_j(V_b D \vec{y})\,.
\]
Left-multiplying by~$C U_a X$ yields the claimed equality.
\end{proof}

The following lemma says that if two matrices $\widehat{X}, X \in S$ have the same $\Psi$-values in the first $b-1$ blocks of block row~$a$, then their $\psi_{a b}$-values are related by an element of the group~$H_b$.
This lemma motivates our term ``coupling group'' for~$H_b$.
\begin{lemma} \label{lem:ratio}
Suppose that $\widehat{X}, X \in S$ satisfy $\psi_{a j}(\widehat{X}) = \psi_{a j}(X)$ for all $1 \le j \le b-1$.
If $\psi_{a b}(\widehat{X}), \psi_{a b}(X) \in G$ (i.e., are nonzero), then there is an $h \in H_b$ such that $\psi_{a b}(\widehat{X}) h = \psi_{a b}(X)$.
\end{lemma}
\begin{proof}
Write $\widehat{g} \coloneqq \psi_{a b}(\widehat{X}) \in G$ and $g \coloneqq \psi_{a b}(X) \in G$.
By \cref{lem:identity-on-Lb},
\[
 \widehat{g} \vec{y} \ = \ \sum_{j=1}^{b-1} \psi_{a j}(\widehat{X}) \Theta_j(V_b D \vec{y}) \ = \ \sum_{j=1}^{b-1} \psi_{a j}(X) \Theta_j(V_b D \vec{y}) \ = \ g \vec{y} \qquad \text{for all $\vec{y} \in \cL_b$\,;}
\]
i.e., $\widehat{g}$ and~$g$ agree on~$\cL_b$.
It follows that $h \coloneqq \widehat{g}^{-1} g$ (where $h \in G$, as $G$~is a group) fixes~$\cL_b$; i.e., $h \vec{y} = \vec{y}$ for all $\vec{y} \in \cL_b$.
Thus, $h \in H_b$ and $\widehat{g} h = g$.
\end{proof}

\begin{example} \label{ex:5}
We continue \cref{ex:4}.
Using the expressions for $C_i D_j$ ($i,j \in \{1,2\}$) from \cref{ex:1} and the fact that $G$~is a group one can show that
\[
\left\{\begin{pmatrix}\psi_{1 1}(X) & \psi_{1 2}(X) \end{pmatrix} \mid X \in S\right\}
\ = \
\left\{\begin{pmatrix} g & g \big(\begin{smallmatrix}0&1\\1&0\end{smallmatrix}\big) \end{pmatrix} \mid g \in G\right\} \; \cup \;
\left\{\begin{pmatrix} g & g \big(\begin{smallmatrix}0&-1\\1&0\end{smallmatrix}\big) \end{pmatrix} \mid g \in G\right\}\,.
\]
Therefore, for any $\widehat{X}, X \in S$ with $\psi_{1 1}(\widehat{X}) = \psi_{1 1}(X)$ we have
\[
 \psi_{1 2}(\widehat{X}) \big(\begin{smallmatrix}1&0\\0&1\end{smallmatrix}\big) \ = \ \psi_{1 2}(X) \qquad \text{or} \qquad
 \psi_{1 2}(\widehat{X}) \big(\begin{smallmatrix}1&0\\0&-1\end{smallmatrix}\big) \ = \ \psi_{1 2}(X)\,.
\]
Since we have $H_2 = \left\{\left(\begin{smallmatrix}1&0\\0&1\end{smallmatrix}\right), \left(\begin{smallmatrix}1&0\\0&-1\end{smallmatrix}\right)\right\}$
from \cref{ex:4}, this matches \cref{lem:ratio}.
\exampleend
\end{example}

The following lemma bounds the number of different $\psi_{a b}(X)$ when the $\psi_{a j}(X)$ for $j < b$ have been fixed.
\begin{lemma}\label{lem:row-step}
Let $g_1, \ldots, g_{b-1} \in G \cup \{O_r\}$.
Then
\[
 |\{\psi_{a b}(X) \mid X \in S, \ \psi_{a j}(X) = g_j \text{ for all } 1 \le j \le b-1\}| \ \le \ 3^{w_b^2}\,.
\]
\end{lemma}
Here is an illustration of the lemma.
\begin{center}
\scalebox{.9}{
\begin{tikzpicture}[x=7mm,y=7mm]
  \fill[darkblue] (\Xzero,\Ytop) rectangle (\Xone,\Yone);   
  \fill[darkblue] (\Xone,\Ytop) rectangle (\Xtwo,\Yone);    
  \fill[myyellow] (\Xtwo,\Ytop) rectangle (\Xthree,\Yone);  

  \node at ({0.5*(\Xzero+\Xone)}, {0.5*(\Ytop+\Yone)}) {$g_1$};        
  \node at ({0.5*(\Xone+\Xtwo)},  {0.5*(\Ytop+\Yone)}) {$g_2$};        
  \node at ({0.5*(\Xtwo+\Xthree)},{0.5*(\Ytop+\Yone)}) {$\psi_{a b}(X)$}; 

  \draw[very thick] (\Xzero,\Ytop) -- (\Xfive,\Ytop);   
  \draw            (\Xzero,\Yone) -- (\Xfive,\Yone);    
  \draw[very thick] (\Xzero,\Ytop) -- (\Xzero,\Yone);   
  \draw[very thick] (\Xfive,\Ytop) -- (\Xfive,\Yone);   
  \foreach \XX in {\Xone,\Xtwo,\Xthree,\Xfour} \draw (\XX,\Ytop) -- (\XX,\Yone);

  \path ({0.5*(\Xzero+\Xone)},0.35)  node {$w_1=r$};
  \path ({0.5*(\Xone+\Xtwo)},0.35)   node {$w_2$};
  \path ({0.5*(\Xtwo+\Xthree)},0.35) node {$w_b$};
  \path ({0.5*(\Xthree+\Xfour)},0.35)node {$w_4$};
  \path ({0.5*(\Xfour+\Xfive)},0.35) node {$w_v$};

  \draw[decorate,decoration={brace,amplitude=6pt}]
    (\Xzero,0.6) -- node[above=8pt] {$n$} (\Xfive,0.6);
\end{tikzpicture}
}
\end{center}
\begin{proof}[Proof of \cref{lem:row-step}]
Set
\[
R \ \coloneqq \ \{\psi_{a b}(X) \mid X \in S,\ \psi_{a j}(X)=g_j \text{ for } 1\le j\le b-1\} \ \subseteq \ G\cup\{O_r\}.
\]
Suppose that $R \setminus \{O_r\}$~is nonempty; i.e., there is $\widehat{X} \in S$ with $\psi_{a j}(\widehat{X}) = g_j$ for all $1 \le j \le b-1$ and $\psi_{a b}(\widehat{X}) \ne O_r$.
Then, by \cref{lem:ratio}, $R \setminus \{O_r\} \subseteq \psi_{a b}(\widehat{X}) H_b$.
It follows that $|R \setminus \{O_r\}| \le |H_b|$.
Thus, $|R| \le |H_b| + 1$.
Clearly, this bound also holds when $R \setminus \{O_r\}$ is empty.
Hence, \cref{lem:Hb-bound} implies $|R| \le 3^{w_b^2}$.
\end{proof}

The following proposition bounds the number of length-$b$ prefixes of the $a$th block row of~$\Psi$.
It will not be used later; it serves as ``warm-up'' for the ``2-dimensional'' \cref{lem:grid-bound} in \cref{sub:overall} below.
\begin{proposition} \label{prop:row-bound} 
Let
$Y_b \coloneqq \left\{\begin{pmatrix} \psi_{a 1}(X) & \cdots & \psi_{a b}(X) \end{pmatrix} \mid X \in S\right\}$.
We have $|Y_b| \le 3^{w_1^2 + \cdots + w_b^2}$.
\end{proposition}
\begin{proof}
The value of $b \in \{1, \ldots, v\}$ was fixed at the beginning of \cref{sub:width}.
In the following we let~$b$ vary, and prove the proposition by induction on $b \in \{1, \ldots, v\}$.
The induction base, $b=1$, follows immediately from \cref{lem:row-step}.
For the induction step, suppose $|Y_{b-1}| \le 3^{w_1^2 + \cdots + w_{b-1}^2}$ holds for some $1 < b \le v$.
We have
\begin{align*}
|Y_b| \ 
& = \ \sum_{\slim{(g_1 \; \cdots \; g_{b-1}) \,\in\, Y_{b-1}}} \mathrlap{|\{\psi_{a b}(X) \mid X \in S, \ \psi_{a j}(X) = g_j \text{ for all } 1 \le j \le b-1\}|} \\
& \le \ \sum_{\slim{(g_1 \; \cdots \; g_{b-1}) \,\in\, Y_{b-1}}} 3^{w_b^2} \ = \ |Y_{b-1}| \cdot 3^{w_b^2} && \text{by \cref{lem:row-step}} \\
& \le \ 3^{w_1^2 + \cdots + w_{b-1}^2} \cdot 3^{w_b^2} \ = \ 3^{w_1^2 + \cdots + w_{b}^2} && \text{by the induction hypothesis.} \qedhere 
\end{align*}
\end{proof}

Using \cref{prop:row-bound}, we can improve the bound $|S| \le 3^{r^2 n^2}$ obtained at the end of \cref{sub:Psi}.
Let us write $Y_{a b} \coloneqq Y_b$ for the set~$Y_b$ from \cref{prop:row-bound}, to make its implicit dependence on $a \in \{1, \ldots, u\}$ (fixed at the beginning of the subsection) explicit.
Since $w_j \le r$ and $w_1 + \cdots + w_v = n$ by \cref{eq:wj}, we have $w_1^2 + \cdots + w_v^2 \le r n$.
Then \cref{prop:row-bound} gives $|Y_{a v}| \le 3^{r n}$, and we obtain, using $u \le n$,
\[
|S| \ = \ |\Psi(S)| \ \le \ \prod_{a=1}^{u} |Y_{a v}| \ \le \ \prod_{a=1}^{u} 3^{r n} \ = \ 3^{r n u} \ \le \ 3^{r n^2}\,,
\]
improving on the earlier bound by a factor of~$r$ in the exponent.

In order to improve this bound down to $|S| \le 3^{n^2}$, we need to exploit dependencies between the block rows, in addition to the dependencies within block row~$a$ explored thus far.
Column dependencies are, of course, completely analogous to row dependencies; the remaining challenge is to find a way to couple each block~$(a,b)$ both within its row and its column.

\subsection{The overall count} \label{sub:overall}

In the previous subsection we considered the length-$b$ prefix of the $a$th block row of~$\Psi$
\[
\begin{pmatrix} \psi_{a 1}(X) & \cdots & \psi_{a b}(X) \end{pmatrix} \qquad \text{where $X \in S$.}
\]
Next we wish to formulate the column analogue of \cref{lem:row-step}.
Analogously to the width~$w_b$ of block column~$b$, we define the \emph{height}, $h_a$, of block row~$a$, i.e.,
\[
 h_a \ \coloneqq \ \dim \U_a - \dim \U_{a-1} \qquad \text{where $1 \le a \le u$.}
\]
We have $h_a > 0$, analogously to $w_b > 0$ from \cref{lem:lw}.
The following equalities are exactly analogous to \cref{eq:wj} for~$w_b$ in \cref{sub:width}:
\begin{equation} \label{eq:hi} \begin{aligned}
 h_1 \ &= \ \dim \U_1 \ = \ \rk(C U_1) \ = \ r    \\
 h_1 + \cdots + h_u \ &= \ \dim \U_u \ = \ n \,.
\end{aligned}
\end{equation}
In particular, $h_1 = \rk(C U_1) = r$ follows from the analogue of \cref{lem:rk-VjD}.
The following lemma considers a length-$a$ prefix of the $b$th block column of~$\Psi$.
\begin{lemma}\label{lem:column-step}
Let $1 \le a \le u$ and $1 \le b \le v$.
Let $g_1, \ldots, g_{a-1} \in G \cup \{O_r\}$.
Then
\[
 |\{\psi_{a b}(X) \mid X \in S, \ \psi_{i b}(X) = g_i \text{ for all } 1 \le i \le a-1\}| \ \le \ 3^{h_a^2}\,.
\]
\end{lemma}
The proof follows from transposing the row argument from the last two subsections; we omit the proof, as it is fully analogous to the proof of \cref{lem:row-step}.

Towards the overall count, define the \emph{grid}
\[
\Gamma \ \coloneqq \ \{1, \ldots, u\} \times \{1, \ldots, v\}\,.
\]
Let $\mathord{\prec}$ be the row-major order on~$\Gamma$, i.e.,
\[
 (i,j) \prec (i',j') \quad \Longleftrightarrow \quad i < i'\quad\text{or}\quad (i = i' \text{ and } j < j').
\]
\begin{figure}[t]
\begin{center}
\begin{tabular}{cc}
    \scalebox{.9}{\begin{tikzpicture}[x=7mm,y=7mm]
      \foreach \xa/\xb in {\Xzero/\Xone,\Xone/\Xtwo,\Xtwo/\Xthree,\Xthree/\Xfour,\Xfour/\Xfive}{
        \fill[lightblue] (\xa,\Ytop) rectangle (\xb,\Yone); 
        \fill[lightblue] (\xa,\Yone) rectangle (\xb,\Ytwo); 
      }
      \foreach \xa/\xb in {\Xzero/\Xone,\Xone/\Xtwo}{
        \fill[lightblue] (\xa,\Ytwo) rectangle (\xb,\Ythree); 
      }

      \draw[very thick] (\Xzero,\Ytop) rectangle (\Xfive,\Yfour);
      \foreach \XX in {\Xone,\Xtwo,\Xthree,\Xfour} \draw (\XX,\Ytop) -- (\XX,\Yfour);
      \foreach \YY in {\Yone,\Ytwo,\Ythree}       \draw (\Xzero,\YY) -- (\Xfive,\YY);

      \path ({0.5*(\Xzero+\Xone)},0.35)  node {$w_1=r$};
      \path ({0.5*(\Xone+\Xtwo)},0.35)   node {$w_2$};
      \path ({0.5*(\Xtwo+\Xthree)},0.35) node {$w_3$};
      \path ({0.5*(\Xthree+\Xfour)},0.35)node {$w_4$};
      \path ({0.5*(\Xfour+\Xfive)},0.35) node {$w_v$};

      \path (-0.05,{0.5*(\Ytop+\Yone)})   node[anchor=east] {$h_1=r$};
      \path (-0.05,{0.5*(\Yone+\Ytwo)})   node[anchor=east] {$h_2$};
      \path (-0.05,{0.5*(\Ytwo+\Ythree)}) node[anchor=east] {$h_3$};
      \path (-0.05,{0.5*(\Ythree+\Yfour)})node[anchor=east] {$h_u$};

      \draw[decorate,decoration={brace,amplitude=6pt}]
        (\Xzero,0.6) -- node[above=8pt] {$n$} (\Xfive,0.6);
      \draw[decorate,decoration={brace,amplitude=6pt,mirror}]
        (-2,\Ytop) -- node[left=8pt] {$n$} (-2,\Yfour);
    \end{tikzpicture}}
&
  \scalebox{.9}{
    \begin{tikzpicture}[x=7mm,y=7mm]
      \foreach \xa/\xb in {\Xzero/\Xone,\Xone/\Xtwo,\Xtwo/\Xthree,\Xthree/\Xfour,\Xfour/\Xfive}{
        \fill[lightblue] (\xa,\Ytop) rectangle (\xb,\Yone); 
        \fill[lightblue] (\xa,\Yone) rectangle (\xb,\Ytwo); 
      }
      \foreach \xa/\xb in {\Xzero/\Xone,\Xone/\Xtwo}{
        \fill[lightblue] (\xa,\Ytwo) rectangle (\xb,\Ythree); 
      }

      \fill[darkblue] (\Xzero,\Ytwo)   rectangle (\Xone,\Ythree);  
      \fill[darkblue] (\Xone,\Ytwo)    rectangle (\Xtwo,\Ythree);  
      \fill[darkblue] (\Xtwo,\Yone)    rectangle (\Xthree,\Ytwo);  
      \fill[darkblue] (\Xtwo,\Ytop)    rectangle (\Xthree,\Yone);  

      \fill[myyellow] (\Xtwo,\Ytwo) rectangle (\Xthree,\Ythree);
      \node at ({0.5*(\Xtwo+\Xthree)}, {0.5*(\Ytwo+\Ythree)}) {$\psi_{a b}(X)$};

      \node at ({0.5*(\Xzero+\Xone)},  {0.5*(\Ytop+\Yone)})   {$g_{1 1}$};
      \node at ({0.5*(\Xone+\Xtwo)},   {0.5*(\Ytop+\Yone)})   {$g_{1 2}$};
      \node at ({0.5*(\Xtwo+\Xthree)}, {0.5*(\Ytop+\Yone)})   {$g_{1 b}$};
      \node at ({0.5*(\Xthree+\Xfour)},{0.5*(\Ytop+\Yone)})   {$g_{1 4}$};
      \node at ({0.5*(\Xfour+\Xfive)}, {0.5*(\Ytop+\Yone)})   {$g_{1 v}$};
      \node at ({0.5*(\Xzero+\Xone)},  {0.5*(\Yone+\Ytwo)})   {$g_{2 1}$};
      \node at ({0.5*(\Xone+\Xtwo)},   {0.5*(\Yone+\Ytwo)})   {$g_{2 2}$};
      \node at ({0.5*(\Xtwo+\Xthree)}, {0.5*(\Yone+\Ytwo)})   {$g_{2 b}$};
      \node at ({0.5*(\Xthree+\Xfour)},{0.5*(\Yone+\Ytwo)})   {$g_{2 4}$};
      \node at ({0.5*(\Xfour+\Xfive)}, {0.5*(\Yone+\Ytwo)})   {$g_{2 v}$};
      \node at ({0.5*(\Xzero+\Xone)},  {0.5*(\Ytwo+\Ythree)}) {$g_{a 1}$};
      \node at ({0.5*(\Xone+\Xtwo)},   {0.5*(\Ytwo+\Ythree)}) {$g_{a 2}$};

      \draw[very thick] (\Xzero,\Ytop) rectangle (\Xfive,\Yfour);
      \foreach \XX in {\Xone,\Xtwo,\Xthree,\Xfour} \draw (\XX,\Ytop) -- (\XX,\Yfour);
      \foreach \YY in {\Yone,\Ytwo,\Ythree}       \draw (\Xzero,\YY) -- (\Xfive,\YY);

      \path ({0.5*(\Xzero+\Xone)},0.35)  node {$w_1=r$};
      \path ({0.5*(\Xone+\Xtwo)},0.35)   node {$w_2$};
      \path ({0.5*(\Xtwo+\Xthree)},0.35) node {$w_b$};
      \path ({0.5*(\Xthree+\Xfour)},0.35)node {$w_4$};
      \path ({0.5*(\Xfour+\Xfive)},0.35) node {$w_v$};

      \path (-0.05,{0.5*(\Ytop+\Yone)})   node[anchor=east] {$h_1=r$};
      \path (-0.05,{0.5*(\Yone+\Ytwo)})   node[anchor=east] {$h_2$};
      \path (-0.05,{0.5*(\Ytwo+\Ythree)}) node[anchor=east] {$h_a$};
      \path (-0.05,{0.5*(\Ythree+\Yfour)})node[anchor=east] {$h_u$};

      \draw[decorate,decoration={brace,amplitude=6pt}]
        (\Xzero,0.6) -- node[above=8pt] {$n$} (\Xfive,0.6);
    \end{tikzpicture}}
\\
(a) & (b)
\end{tabular}
\end{center}
\caption{(a) Grid~$\Gamma$ with the cells $\prec (3,3)$ shown in blue.
(b) Illustration of the proof of \cref{lem:grid-step}; 
the number of possible values~$\psi_{a b}(X)$ in the yellow cell is limited by the possible combinations of values in the dark-blue cells.}
\label{fig:rowmajor-gridstep}
\end{figure}
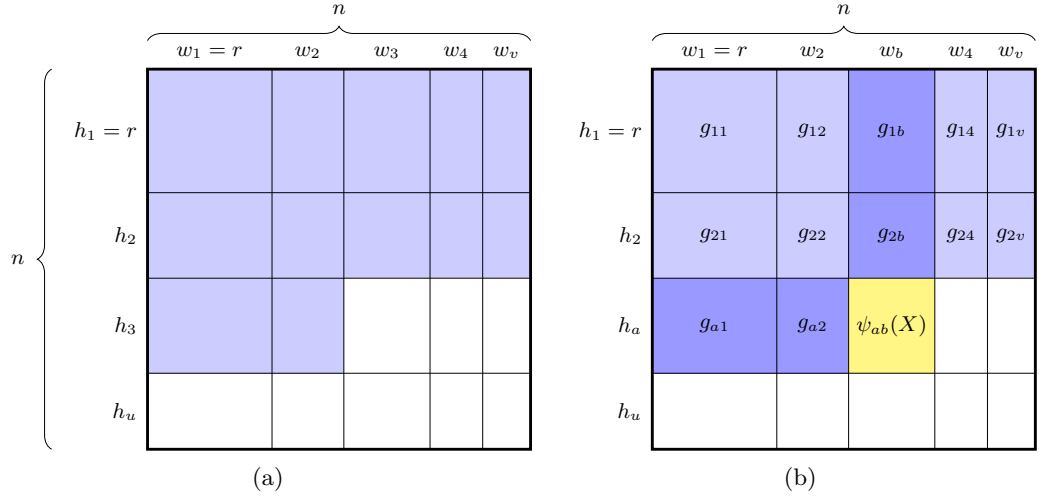
\Cref{fig:rowmajor-gridstep}(a) visualizes the order~$\mathord{\prec}$.

The following lemma is a grid analogue to \cref{lem:row-step,lem:column-step}; in fact, the proof is based on these lemmas.
\begin{lemma} \label{lem:grid-step}
Let $(a,b) \in \Gamma$.
For all $(i,j) \prec (a,b)$ fix $g_{i j} \in G \cup \{O_r\}$.
Let
\[
 R \ \coloneqq \ \{\psi_{a b}(X) \mid X \in S, \ \psi_{i j}(X) = g_{i j} \text{ for all } (i,j) \prec (a,b)\}\,.
\]
Then $|R| \le 3^{h_a w_b}$.
\end{lemma}
\begin{proof}
Note that $(a,j) \prec (a,b)$ and $(i,b) \prec (a,b)$ for all $1 \le j \le b-1$ and all $1 \le i \le a-1$.
\Cref{fig:rowmajor-gridstep}(b) shows these grid elements in dark-blue.
We have
\begin{align*}
R \ & \subseteq \ \{\psi_{a b}(X) \mid X \in S, \ \psi_{a j}(X) = g_{a j} \text{ for all } 1 \le j \le b-1\} \qquad \text{and}\\
R \ & \subseteq \ \{\psi_{a b}(X) \mid X \in S, \ \psi_{i b}(X) = g_{i b} \text{ for all } 1 \le i \le a-1\}\,.
\end{align*}
Using \cref{lem:row-step,lem:column-step} respectively, we obtain $|R| \le 3^{w_b^2}$ and $|R| \le 3^{h_a^2}$.
Since $h_a, w_b > 0$,
\[
|R| \ \le \ \min\left\{3^{h_a^2},3^{w_b^2}\right\} \ = \ 3^{(\min\{h_a,w_b\})^2} \ \le \ 3^{h_a w_b}\,. \qedhere
\]
\end{proof}

The following lemma is the grid analogue to \cref{prop:row-bound}.
\begin{lemma} \label{lem:grid-bound}
Let $1 \le k \le u v$ and let $(a,b) \in \Gamma$ be the $k$th~pair in the order~$\mathord{\prec}$. 
Define
\[
Z_k \ \coloneqq\ \Bigl\{\big(\psi_{i j}(X)\big)_{(i,j) \preceq (a,b)}\ \Bigm|\ X \in S\Bigr\}\,.
\]
Set $s \coloneqq \sum_{(i,j) \preceq (a,b)} h_i w_j$.
Then we have $|Z_k| \le 3^s$.
\end{lemma}
\begin{proof}
We prove the lemma by induction on $k \in \{1, \ldots, u v\}$.
The induction base, $k=1$, follows immediately from \cref{lem:grid-step}.
For the induction step, let $1 < k \le u v$, and suppose that $|Z_{k-1}| \le 3^{s_0}$ holds for $s_0 \coloneqq \sum_{(i,j) \prec (a,b)} h_i w_j$, where $(a,b) \in \Gamma$ is the $k$th~pair in the order~$\mathord{\prec}$.
\SKT{For each tuple in~$Z_{k-1}$, \cref{lem:grid-step} gives at most $3^{h_a w_b}$ possible values of the next block~$\psi_{a b}(X)$.
Therefore,}
\[
\SKT{|Z_k| \ \le \ |Z_{k-1}| \cdot 3^{h_a w_b}
\ \le \ 3^{s_0} \cdot 3^{h_a w_b} \ = \ 3^s\,.} \qedhere
\]
\end{proof}

Now the main theorem follows.

\begin{theorem}\label{thm:main-bound}
Let $S \subseteq \Q^{n \times n}$ be a finite irreducible semigroup. Then $|S| \le 3^{n^2}$.
\end{theorem}
\begin{proof}
Recall from \cref{eq:wj,eq:hi} that $h_1 + \cdots + h_u = n = w_1 + \cdots + w_v$.
We have
\begin{align*}
|S| \ 
&= \ |\Psi(S)| && \text{as $\Psi$~is injective by \cref{lem:Psi-injective}} \\
&= \ |Z_{u v}| \ \le \ 3^s && \text{by \cref{lem:grid-bound}}\,,
\end{align*}
where
$
s  =  \sum_{(i,j) \in \Gamma} h_i w_j  =  \sum_{i=1}^{u} h_i \sum_{j=1}^v w_j  =  n \cdot n  =  n^2\,.
$
\end{proof}

\SKT{Both our proof and that of Steinberg~\cite{Steinberg26} use the group associated with an idempotent in the \mbox{(0-)}minimal ideal and Minkowski's theorem for finite matrix groups.
Steinberg first changes basis so that the semigroup consists of integer matrices.
He then uses results from the representation theory of finite semigroups to show that reducing entries modulo~$3$ gives an injective map on the whole semigroup, yielding the bound~$3^{n^2}$.
Our proof instead counts the possible values of~$\Psi(X)$.
It uses the overlaps between the spaces $\im(V_j D)$, and the resulting groups~$H_b$, to bound the possible values of successive blocks.
Thus our argument describes explicitly how these overlaps constrain the blocks used in the count.}

\section{Mortality threshold and minimum-rank diameter}\label{sec:mortality}
In this section we prove the following theorem.

\begin{theorem} \label{thm:min-rank-diameter}
Let $S \subseteq \Q^{n \times n}$ be a finite, not necessarily irreducible, semigroup, generated by $S_0 \subseteq S$.
If $S$ contains the zero matrix, then its mortality threshold is at most~$3^{n^2}$.
\end{theorem}

Roughly speaking, \Cref{thm:min-rank-diameter} is proved by decomposing~$S$ into irreducible ``parts'' and using \Cref{thm:main-bound}.

Below we prove a result that is slightly stronger than \cref{thm:min-rank-diameter}.
To state it, we define the \emph{minimum-rank diameter} of~$S$ as the \emph{minimum} depth among the minimum-rank (possibly zero) matrices $X \in S$.
First we note the following simple fact.
\begin{proposition} \label{prop:diameter-size-bound}
Let $S \subseteq \Q^{n \times n}$ be a finite, not necessarily irreducible, semigroup, generated by $S_0 \subseteq S$.
For any $X \in S$, its depth is at most~$|S|$.
Hence, the minimum-rank diameter is at most~$|S|$.
\end{proposition}
\begin{proof}
\SKT{Let $M : \Sigma^+ \to S$ be a semigroup homomorphism with $M(\Sigma) = S_0$ and $M(\Sigma^+) = S$.}
Let $\ell$~be the depth of~$X$, so that $X = M(a_1 \cdots a_\ell)$ for some $a_1, \ldots, a_\ell \in \Sigma$.
Then the ``prefix products'' $M(a_1 \cdots a_k)$ with $1 \le k \le \ell$ are all distinct: indeed, if $M(a_1 \cdots a_i) = M(a_1 \cdots a_j)$ for some $i < j$, then
\[
 M(a_1 \cdots a_i a_{j+1} \cdots a_\ell) \ = \ M(a_1 \cdots a_i) M(a_{j+1} \cdots a_\ell) \ = \ M(a_1 \cdots a_j) M(a_{j+1} \cdots a_\ell) \ = \ X\,,
\]
contradicting that $\ell$~is the depth of~$X$.
Thus, $\ell \le |S|$.
\end{proof}

If the zero matrix is in the semigroup, then the zero matrix is the unique minimum-rank matrix and, thus, the mortality threshold equals the minimum-rank diameter.
Therefore, the following proposition implies \cref{thm:min-rank-diameter}.

\begin{proposition} \label{prop:min-rank-diameter}
Let $S \subseteq \Q^{n \times n}$ be a finite, not necessarily irreducible, semigroup, generated by $S_0 \subseteq S$.
Then its minimum-rank diameter is at most~$3^{n^2}$.
\end{proposition}
\begin{proof}
We prove the proposition by (strong) induction on~$n \ge 1$.
Let the finite semigroup $S \subseteq \Q^{n \times n}$ be generated by~$S_0 \subseteq S$.
Let $M : \Sigma^+ \to S$ be a semigroup homomorphism with $M(\Sigma) = S_0$ and $M(\Sigma^+) = S$.

Suppose $S$~is irreducible.
By \cref{thm:main-bound} we have $|S| \le 3^{n^2}$.
Hence, by \cref{prop:diameter-size-bound} the minimum-rank diameter is at most~$3^{n^2}$.

So we can assume that $S$~is not irreducible.
Choose~$\V_1$ to be a minimal nonzero $S$-invariant subspace of~$\Q^n$; then $\{\vec{0}\} \ne \V_1 \ne \Q^n$.
Let $\V_2 \subseteq \Q^n$ be a complement of~$\V_1$ \SKT{(not necessarily $S$-invariant)}, so that $\Q^n = \V_1 \oplus \V_2$.
Let $n_i \coloneqq \dim \V_i$ for $i=1,2$.
We have $n = n_1 + n_2$ with $n_1,n_2 \ge 1$.
Since $\V_1$~is $S$-invariant, in a basis adapted to $\Q^n = \V_1 \oplus \V_2$ we have
\[
M(w) \ = \ \begin{pmatrix} M_1(w) & N(w) \\ 0 & M_2(w) \end{pmatrix} \qquad \text{for all $w \in \Sigma^+$,}
\]
where $M_i(w) \in \Q^{n_i \times n_i}$ for $i = 1,2$, and $N(w) \in \Q^{n_1 \times n_2}$.

Define $S_i \coloneqq M_i(\Sigma^+)$ for $i=1,2$.
Then for $i=1,2$ the set $S_i \subseteq \Q^{n_i \times n_i}$ is a finite semigroup generated by~$M_i(\Sigma)$.
The semigroup $S_1$ is irreducible; indeed, identifying $\Q^{n_1}$ with~$\V_1$ via the chosen basis, any proper nonzero $S_1$-invariant subspace of~$\Q^{n_1}$ would correspond to a proper nonzero $S$-invariant subspace of~$\V_1$, contradicting the minimality of~$\V_1$.

It follows from the block-triangular shape above that $\rk M(w) \ge \rk M_1(w) + \rk M_2(w)$ for all $w \in \Sigma^+$.
In particular, defining $r \coloneqq \min\{\rk X \mid X \in S\}$ and $r_i \coloneqq \min\{\rk X \mid X \in S_i\}$ for $i=1,2$, we have $r \ge r_1 + r_2$.
For $i=1,2$ let $w_i \in \Sigma^+$ be shortest words with $\rk M_i(w_i) = r_i$.

Since $S_1$~is irreducible, by \cref{thm:main-bound} we have $|S_1| \le 3^{n_1^2}$, hence by \cref{prop:diameter-size-bound}, $|w_1| \le 3^{n_1^2}$.
Applying the induction hypothesis to~$S_2$ we obtain $|w_2| \le 3^{n_2^2}$.
Further,
\begin{align*}
M(w_1 w_2) \ &= \ M(w_1) M(w_2) \ = \ 
\begin{pmatrix}
    M_1(w_1) & N(w_1) \\ 0 & M_2(w_1)
\end{pmatrix}
\begin{pmatrix}
    M_1(w_2) & N(w_2) \\ 0 & M_2(w_2)
\end{pmatrix} \\
&= \ 
\underbrace{\begin{pmatrix}
    M_1(w_1) \\ 0 
\end{pmatrix}}_{\text{rank } r_1}
\begin{pmatrix}
    M_1(w_2) & N(w_2)
\end{pmatrix}
+
\begin{pmatrix}
    N(w_1) \\ M_2(w_1)
\end{pmatrix}
\underbrace{
\begin{pmatrix}
   0 & M_2(w_2)
\end{pmatrix}
}_{\text{rank } r_2}\,.
\end{align*}
Thus, using $\rk(X Y) \le \min \{\rk X, \rk Y\}$, the first summand has rank $\le r_1$, and the second summand has rank $\le r_2$.
Using $\rk(X + Y) \le \rk X + \rk Y$, we conclude that $\rk M(w_1 w_2) \le r_1 + r_2$.
Since $r_1 + r_2 \le r$, the matrix $M(w_1 w_2)$ is a minimum-rank matrix in~$S$.
Hence, the minimum-rank diameter is at most
\[
  |w_1 w_2| \ = \ |w_1| + |w_2| \ \le \ 3^{n_1^2} + 3^{n_2^2} \ \le \ 2 \cdot 3^{\max\{n_1^2, n_2^2\}} \ \le \ 3^{n_1^2 + n_2^2} \ \le \ 3^{(n_1+n_2)^2} \ = \ 3^{n^2}\,. \qedhere
\]
\end{proof}

\section{Lower bound} \label{sec:lower}
Fix $n \ge 2$ and write $n = p+q$ with
\[
  p \ \coloneqq \ \left\lfloor \tfrac{n}{2} \right\rfloor, \qquad 
  q \ \coloneqq \ \big\lceil \tfrac{n}{2} \big\rceil, \qquad
  P \ \coloneqq \ \{1,\ldots,p\}, \qquad 
  Q \ \coloneqq \ \{p+1,\ldots,n\}\,.
\]
We view $P$ as the ``north-west'' index set and $Q$ as the ``south-east'' one; note $P \cap Q = \emptyset$. 
For $X \in \Z^{n\times n}$ we write
\[
  \supp X \ \coloneqq \ \{(i,j)\in\{1,\dots,n\}^2 \mid X_{ij}\neq 0\}
\]
for the support of~$X$.
We define four families of matrices with entries in  $\{-1, 0, 1\}$:
\begin{align*}
\NE   \ &\coloneqq\ \{X \in \{-1,0,1\}^{n\times n} \mid \supp X \subseteq P\times Q\}\,,\\
\COL  \ &\coloneqq\ \{X \in \{-1,0,1\}^{n\times n} \mid \exists\, b \in P:\ \supp X \subseteq P \times \{b\}\}\,,\\
\ROW  \ &\coloneqq\ \{X \in \{-1,0,1\}^{n\times n} \mid \exists\, a \in Q:\ \supp X \subseteq \{a\} \times Q\}\,,\\
\UNIT \ &\coloneqq\ \{X \in \{-1,0,1\}^{n\times n} \mid |\supp X| \le 1\}\,.
\end{align*}
In words, $\NE$~consists of the north-east-supported matrices; $\COL$~consists of the north-west-supported matrices with at most one nonzero column; $\ROW$ consists of the south-east-supported matrices with at most one nonzero row; and $\UNIT$~consists of the signed matrix units and~$0$.
Each family includes the zero matrix.
The following proposition complements \cref{thm:main-bound}.
\begin{proposition} \label{prop:lower}
\SKT{Let $S \coloneqq \NE \cup \COL \cup \ROW \cup \UNIT$.}
The set $S \subseteq \{-1,0,1\}^{n \times n}$ is a finite irreducible integer matrix semigroup with at least $3^{\lfloor n^2/4 \rfloor}$~elements.
\end{proposition}
\begin{proof}
Clearly, $S$~is finite.
To argue that $S$~is closed under multiplication, consider the following multiplication table.
\[
\begin{array}{c|cccc}
\mathord{\cdot}                   & \NE           & \COL  & \ROW          & \UNIT         \\ \hline
\NE                               & \{O_n\}         & \{O_n\} & \NE           & \NE \cup \COL \\
\COL                              & \NE           & \COL  & \{O_n\}         & \COL \cup \NE \\
\ROW                              & \{O_n\}         & \{O_n\} & \ROW          & \UNIT \\
\UNIT                             & \NE \cup \ROW & \UNIT & \NE \cup \ROW & \UNIT
\end{array}
\]
For example, the entry in row~$\NE$ and column~$\ROW$ is~$\NE$, to indicate that $\NE \cdot \ROW \subseteq \NE$.
To show this, let $X \in \NE$ and $Y \in \ROW$.
Since $Y \in \ROW$, there is $a \in Q$ with $\supp Y \subseteq \{a\} \times Q$.
Moreover, $\supp X \subseteq P \times Q$.
It follows that $\supp (X Y) \subseteq P \times Q$; i.e., for $(i,j) \not\in P \times Q$ we have $(X Y)_{i j} = 0$.
For any $(i,j) \in P \times Q$,
\[
 (X Y)_{i j} \ = \ \sum_{k=1}^n X_{i k} Y_{k j} \ = \ X_{i a} Y_{a j} \ \in \ \{-1,0,1\}\,.
\]
It follows that $X Y \in \NE$.
The rest of the multiplication table above is shown similarly.
In particular, in every product the support constraints force the summation $(X Y)_{i j} = \sum_{k} X_{i k} Y_{k j}$ to have at most one nonzero term; so the entries remain in~$\{-1,0,+1\}$.

For irreducibility, let $\{\vec{0}\} \ne \V \subseteq \Q^n$ be $S$-invariant.
Since $\V \ne \{\vec{0}\}$ and $\V$~is closed under scalar multiplication, there is $\vec{v} \in \V$ with ${\vec v}_j = 1$ for some $1 \le j \le n$.
Let $1 \le i \le n$.
It suffices to show that $\vec{e}_i \in \V$, where $\vec{e}_i \in \{0,1\}^n$ denotes the $i$th coordinate vector.
To that end, let $E_{i j} \in \UNIT \subseteq S$ be the matrix whose only nonzero entry is a~$1$ at position $(i,j)$.
Then $\vec{e}_i = E_{i j} \vec{v} \in \V$, as $\V$~is $S$-invariant.

Finally,
$
|S| \ge |\NE| = \left|\{-1,0,1\}^{P \times Q}\right| = 3^{|P| |Q|} = 3^{\lfloor n^2/4\rfloor}
$.
\end{proof}


The following proposition complements \cref{thm:aper-bound}.
\begin{proposition} \label{prop:lower-nonnegative}
\SKT{Let $S$ be the semigroup from \cref{prop:lower}.}
\SKT{Then the} set $S \cap \{0,1\}^{n \times n}$ is an aperiodic irreducible semigroup of matrices with entries in $\{0, 1\}$. It has at least $2^{\lfloor n^2/4 \rfloor}$~elements.
\end{proposition}
\begin{proof}
Define $S_{\ge 0} \coloneqq S \cap \{0,1\}^{n \times n}$.
Since $S$~is closed, it follows that $S_{\ge 0}$~is closed; i.e., $S_{\ge 0}$ is a semigroup.
The element count and the irreducibility argument from \cref{prop:lower} carry over to~$S_{\ge 0}$ analogously.
It remains to show that $S_{\ge 0}$ is aperiodic.

Let $K$~be a subgroup of~$S_{\ge 0}$ with identity~$E$.
Then $E E = E$; i.e., $E$~is idempotent.
We need to show that $K = \{E\}$.
If $O_n \in K$ then $K = \{O_n\} = \{E\}$.
So we assume that $O_n \not\in K$; in particular, $E \ne O_n$.

If $E \in \NE$, then $E = E E = O_n$, contradicting our assumption.

Suppose $E \in \COL \cup \UNIT$.
Then $E$~is supported on some column $b \in \{1, \ldots, n\}$; i.e., writing $\vec{b} \in \{0,1\}^n$ for the $b$th coordinate vector, we have $E = \vec{e} \vec{b}^\top$ for some $\vec{e} \in \{0,1\}^n$.
Let $X \in K$.
Since $E$~is the identity in~$K$,  we have
\[
 X \ = \ E X \ = \ E X E \ = \ (\vec{e} \vec{b}^\top) X (\vec{e} \vec{b}^\top) \ = \ \vec{e} (\vec{b}^\top X \vec{e}) \vec{b}^\top \ = \ (\vec{b}^\top X \vec{e}) \vec{e} \vec{b}^\top \ = \ (\vec{b}^\top X \vec{e}) E\,;
\]
i.e., $X \in \{0,1\}^{n \times n} \setminus \{O_n\}$ is a nonnegative integer multiple of~$E$.
It follows that $X = E$.
Since $X \in K$ was arbitrary, we conclude that $K = \{E\}$.

\SKT{Suppose $E \in \ROW$.
Then $E = \vec{a}\vec{e}^{\top}$, where $\vec{a}$ is a coordinate vector and $\vec{e} \in \{0,1\}^n$.
For every $X \in K$, since $E$ is the identity of~$K$,}
\[
 \SKT{X \ = \ E X E \ = \ \vec{a}(\vec{e}^{\top}X\vec{a})\vec{e}^{\top}
 \ = \ (\vec{e}^{\top}X\vec{a})E.}
\]
\SKT{Thus, $X$ is a nonnegative integer multiple of~$E$.
Since $X$ and~$E$ are nonzero $0$-$1$ matrices, $X=E$.
Hence, $K=\{E\}$.}
\end{proof}

\subsection{Zero-free irreducible semigroups}\label{sub:lower-zero-free}
The examples from the previous lower-bound constructions all contain the zero matrix. In fact, most matrices in these semigroups are nilpotent, so the presence of the zero matrix might seem essential. We now show that one can still obtain zero-free irreducible semigroups of size $2^{\Omega(n^2)}$.

The argument has two ingredients. We first construct a large zero-free strongly connected semigroup of $0$-$1$ matrices by a block-based construction. We then pass to signed matrices and use strong connectivity to obtain irreducibility.

For a semigroup $S \subseteq \{0,1\}^{n \times n}$, let us write $\Gamma(S)$ for the directed graph on $\{1,\ldots,n\}$ with an edge $p \to q$ if there exists $X \in S$ with $X_{q p} = 1$. We call $S$ \emph{strongly connected} if $\Gamma(S)$ is strongly connected. Since $S$ is a semigroup of $0$-$1$ matrices, this is equivalent to requiring that for every $p,q \in \{1,\ldots,n\}$ there exists $X \in S$ with $X_{q p} = 1$.

\SKT{Partition $\{1,\ldots,n\}$ into nonempty sets $B_1,\ldots,B_m$, using the same partition for rows and columns.
For each set of rows~$B_i$, choose one set of columns~$B_{\sigma(i)}$.
In each column indexed by~$B_{\sigma(i)}$, put exactly one~$1$ in a row belonging to~$B_i$.
For each $v \in B_{\sigma(i)}$, the value $f_i(v)$ records the chosen row in column~$v$.
All entries not specified this way are~$0$.
The next proposition shows that products of matrices constructed this way again have this form, and counts the distinct matrices.}

\SKT{We use $\sqcup$ to denote a union of pairwise disjoint sets.}

\begin{proposition}\label{prop:lower-zero-free-strong}
Let $\{1,\ldots,n\} = B_1 \sqcup \cdots \sqcup B_m$ be a partition into nonempty blocks, and write
\[
  b_i \ \coloneqq \ |B_i| \qquad (1 \le i \le m)
\]
for the block sizes. For every function $\sigma : \{1,\ldots,m\} \to \{1,\ldots,m\}$ and every family of functions
\[
  f_i : B_{\sigma(i)} \to B_i \qquad (1 \le i \le m)
\]
we define a matrix $A(\sigma,f) \in \{0,1\}^{n \times n}$ by
\[
  A(\sigma,f)_{u v} = 1
  \iff
  \bigl(u \in B_i,\ v \in B_{\sigma(i)},\ u = f_i(v) \text{ for some } i \in \{1,\ldots,m\}\bigr).
\]
Let $\U(B_1,\ldots,B_m)$ denote the set of all matrices $A(\sigma,f)$.
Then $\U(B_1,\ldots,B_m)$ is a finite zero-free strongly connected semigroup, and
\[
  |\U(B_1,\ldots,B_m)|
  =
  \prod_{i=1}^m \left( \sum_{t=1}^m b_i^{b_t} \right).
\]
\end{proposition}

\begin{example}\label{ex:lower-zero-free-blocks}
\SKT{Take $n=4$, $B_1=\{1,2\}$, $B_2=\{3,4\}$, and}
\[
 \SKT{\sigma(1)=\sigma(2)=2,\qquad
 f_1(3)=1,\quad f_1(4)=2,\qquad
 f_2(3)=f_2(4)=3.}
\]
\SKT{Then}
\[
 \SKT{A=A(\sigma,f)=
 \left(\begin{array}{cc|cc}
 0&0&1&0\\
 0&0&0&1\\ \hline
 0&0&1&1\\
 0&0&0&0
 \end{array}\right),
 \qquad
 A^2=
 \left(\begin{array}{cc|cc}
 0&0&1&1\\
 0&0&0&0\\ \hline
 0&0&1&1\\
 0&0&0&0
 \end{array}\right).}
\]
\SKT{The lines indicate the partition into $B_1,B_2$.
In accordance with the multiplication rule in the proof below, $A^2$ has the same~$\sigma$, since $\sigma\circ\sigma=\sigma$.
Its functions are $f_1\circ f_2$ and $f_2\circ f_2$, which are constant with values~$1$ and~$3$, respectively.}
\exampleend
\end{example}

\begin{proof}[Proof of \cref{prop:lower-zero-free-strong}]
Fix
\[
  A = A(\sigma,f), \qquad B = A(\tau,g)
\]
from $\U(B_1,\ldots,B_m)$. For $u \in B_i$ and $z \in B_t$ we have
\[
  (A B)_{u z}
  =
  \sum_{y=1}^n A_{u y} B_{y z}
  =
  \sum_{j : \tau(j)=t} A_{u, g_j(z)}\,.
\]
Since $g_j(z) \in B_j$, the definition of $A$ shows that at most one summand can be nonzero, namely the one with $j = \sigma(i)$. Hence
\[
  (A B)_{u z} = 1
  \iff
  \bigl(\tau(\sigma(i)) = t \text{ and } u = f_i(g_{\sigma(i)}(z))\bigr).
\]
Therefore $A B = A(\tau \circ \sigma,h)$, where
\[
  h_i \ \coloneqq \ f_i \circ g_{\sigma(i)} : B_{(\tau \circ \sigma)(i)} \to B_i
\]
for all $i \in \{1,\ldots,m\}$. So $\U(B_1,\ldots,B_m)$ is closed under multiplication.

Every matrix in $\U(B_1,\ldots,B_m)$ is nonzero: for any $i \in \{1,\ldots,m\}$ and any $v \in B_{\sigma(i)}$ we have
\[
  A(\sigma,f)_{f_i(v),v} = 1\,.
\]
So the semigroup is zero-free.

To show strong connectivity, fix $p,q \in \{1,\ldots,n\}$. Let $p \in B_t$ and $q \in B_i$. Choose $\sigma$ with $\sigma(i)=t$, and choose $f_i : B_t \to B_i$ with $f_i(p)=q$. Choosing the remaining values of $\sigma$ and the remaining maps $f_j$ arbitrarily yields a matrix $A(\sigma,f)$ with $A(\sigma,f)_{q p}=1$.

It remains to count the matrices. Fix $i \in \{1,\ldots,m\}$. One may first choose $\sigma(i)=t \in \{1,\ldots,m\}$ and then choose an arbitrary function $B_t \to B_i$, of which there are $b_i^{b_t}$. These choices are independent over~$i$, so there are at most
\[
  \prod_{i=1}^m \left( \sum_{t=1}^m b_i^{b_t} \right)
\]
possible matrices.
On the other hand, the matrix $A(\sigma,f)$ determines $\sigma(i)$ and $f_i$ for every~$i$: indeed, the set of columns with a nonzero entry in the rows indexed by~$B_i$ is precisely $B_{\sigma(i)}$, and once $\sigma(i)$ is known, each column $v \in B_{\sigma(i)}$ has exactly one $1$ in the block~$B_i$, namely in row~$f_i(v)$. So the displayed upper bound is attained.
\end{proof}

The next proposition shows that the same block construction remains closed after one freely chooses signs on all nonzero entries.

For $X \in \{0,1\}^{n \times n}$, let
\[
  \Sigma(X)
  \ \coloneqq \ 
  \{Y \in \{-1,0,1\}^{n \times n} \mid \supp Y = \supp X\}
\]
denote the set of all signings of~$X$.
For a semigroup $S \subseteq \{0,1\}^{n \times n}$, we call the semigroup generated by $\bigcup_{X \in S} \Sigma(X)$ the \emph{signed version} of~$S$.

\begin{proposition}\label{prop:lower-zero-free-signed-family}
Let $\{1,\ldots,n\} = B_1 \sqcup \cdots \sqcup B_m$ be as in \cref{prop:lower-zero-free-strong}. For every function $\sigma$ and every family of functions
\[
  f_i : B_{\sigma(i)} \to B_i,
  \qquad
  \varepsilon_i : B_{\sigma(i)} \to \{-1,1\}
  \qquad (1 \le i \le m)
\]
we define a matrix $A(\sigma,f,\varepsilon) \in \{-1,0,1\}^{n \times n}$ by
\[
  A(\sigma,f,\varepsilon)_{u v}
  =
  \begin{cases}
    \varepsilon_i(v) &\text{if } u \in B_i,\ v \in B_{\sigma(i)},\ u = f_i(v) \text{ for some } i,\\
    0 &\text{otherwise.}
  \end{cases}
\]
Let $\widetilde{\U}(B_1,\ldots,B_m)$ denote the set of all matrices $A(\sigma,f,\varepsilon)$. Then $\widetilde{\U}(B_1,\ldots,B_m)$ is a finite zero-free semigroup, it is the signed version of $\U(B_1,\ldots,B_m)$, and
\[
  |\widetilde{\U}(B_1,\ldots,B_m)|
  =
  \prod_{i=1}^m \left( \sum_{t=1}^m (2 b_i)^{b_t} \right).
\]
\end{proposition}

\begin{samepage}
\begin{example}\label{ex:lower-zero-free-signs}
\SKT{With the partition, $\sigma$, and $f_i$ from \cref{ex:lower-zero-free-blocks}, choose}
\[
 \SKT{\varepsilon_1(3)=1,\quad \varepsilon_1(4)=-1,\qquad
 \varepsilon_2(3)=-1,\quad \varepsilon_2(4)=1.}
\]
\SKT{Then}
\[
 \SKT{\widetilde A=A(\sigma,f,\varepsilon)=
 \left(\begin{array}{cc|cc}
 0&0&1&0\\
 0&0&0&-1\\ \hline
 0&0&-1&1\\
 0&0&0&0
 \end{array}\right),
 \qquad
 \widetilde A^2=
 \left(\begin{array}{cc|cc}
 0&0&-1&1\\
 0&0&0&0\\ \hline
 0&0&1&-1\\
 0&0&0&0
 \end{array}\right).}
\]
\SKT{The nonzero entries of $\widetilde A^2$ occur in the same positions as those of $A^2$ in \cref{ex:lower-zero-free-blocks}.}
\exampleend
\end{example}
\end{samepage}

\begin{proof}[Proof of \cref{prop:lower-zero-free-signed-family}]
Let
\[
  A = A(\sigma,f,\varepsilon), \qquad B = A(\tau,g,\delta)
\]
from $\widetilde{\U}(B_1,\ldots,B_m)$. For $u \in B_i$ and $z \in B_t$ we have
\[
  (A B)_{u z}
  =
  \sum_{y=1}^n A_{u y} B_{y z}
  =
  \sum_{j : \tau(j)=t} A_{u, g_j(z)} \cdot \delta_j(z)\,.
\]
Again only the summand with $j = \sigma(i)$ can be nonzero. Hence
\[
  (A B)_{u z}
  =
  \begin{cases}
    \varepsilon_i(g_{\sigma(i)}(z)) \, \delta_{\sigma(i)}(z)
    &\text{if } \tau(\sigma(i)) = t \text{ and } u = f_i(g_{\sigma(i)}(z)),\\
    0 &\text{otherwise.}
  \end{cases}
\]
Therefore $A B = A(\tau \circ \sigma,h,\eta)$, where
\[
  h_i \ \coloneqq \ f_i \circ g_{\sigma(i)}
  \qquad\text{and}\qquad
  \eta_i(z) \ \coloneqq \ \varepsilon_i(g_{\sigma(i)}(z))\, \delta_{\sigma(i)}(z)
\]
for all $i$ and all $z \in B_{(\tau \circ \sigma)(i)}$. So $\widetilde{\U}(B_1,\ldots,B_m)$ is a semigroup.

As in the proof of \cref{prop:lower-zero-free-strong}, every matrix $A(\sigma,f,\varepsilon)$ is nonzero. Indeed, for every $v \in B_{\sigma(i)}$ we have
\[
  A(\sigma,f,\varepsilon)_{f_i(v),v} = \varepsilon_i(v) \in \{-1,1\}.
\]

Every signing of every matrix in $\U(B_1,\ldots,B_m)$ belongs to $\widetilde{\U}(B_1,\ldots,B_m)$ by construction. Conversely, every product of such signings stays in $\widetilde{\U}(B_1,\ldots,B_m)$ by the closure argument above. So $\widetilde{\U}(B_1,\ldots,B_m)$ is exactly the signed version of $\U(B_1,\ldots,B_m)$.

For the cardinality, fix $i \in \{1,\ldots,m\}$. One may first choose $\sigma(i)=t$ and then, for every element of~$B_t$, choose independently a signed image in~$B_i$, meaning a target in~$B_i$ together with a sign. Hence there are $(2 b_i)^{b_t}$ possibilities. As in the proof of \cref{prop:lower-zero-free-strong}, these choices are independent over~$i$, and the resulting parametrization is injective. This yields the displayed formula.
\end{proof}

The relevance of signed versions is that strong connectivity forces irreducibility.

\begin{proposition}\label{prop:lower-zero-free-signed-irreducible}
Let $S \subseteq \{0,1\}^{n \times n}$ be a strongly connected semigroup, and let $\widetilde{S}$ be its signed version. Then $\widetilde{S}$ is irreducible.
\end{proposition}
\begin{proof}
Let $\{\vec{0}\} \ne \V \subseteq \Q^n$ be $\widetilde{S}$-invariant. Choose $\vec{v} \in \V$ nonzero and pick $p \in \{1,\ldots,n\}$ with ${\vec v}_p \ne 0$.
Let $q \in \{1,\ldots,n\}$.
Since $S$ is strongly connected, there exists $X \in S$ with $X_{q p}=1$.
Choose two matrices $Y^+,Y^- \in \Sigma(X)$ that agree in every entry except at position $(q,p)$, where
\[
  Y^+_{q p} = 1, \qquad Y^-_{q p} = -1\,.
\]
Then $Y^+,Y^- \in \widetilde{S}$, and hence
\[
  (Y^+ - Y^-) \vec{v} \in \V\,.
\]
The matrix $Y^+ - Y^-$ has exactly one nonzero entry, namely~$2$ at position $(q,p)$, so
\[
  (Y^+ - Y^-) \vec{v} = 2 {\vec v}_p \, \vec{e}_q\,.
\]
Since ${\vec v}_p \ne 0$, it follows that $\vec{e}_q \in \V$.
As $q$ was arbitrary, we obtain $\V = \Q^n$.
\end{proof}

We can now combine the previous propositions.

\begin{proposition}\label{prop:lower-zero-free}
For every integer $n \ge 1$ there exists a finite zero-free irreducible integer matrix semigroup $S \subseteq \{-1,0,1\}^{n \times n}$ with at least
\[
  2^{2 \lfloor n/2 \rfloor \lfloor n/4 \rfloor}
\]
elements.
\end{proposition}
\begin{proof}
For $n=1$, the singleton $S = \{I_1\}$ has the required properties.
So assume $n \ge 2$.
Let
\[
  N \ \coloneqq \ \left\lfloor \tfrac{n}{2} \right\rfloor,
  \qquad
  k \ \coloneqq \ \left\lfloor \tfrac{n}{4} \right\rfloor.
\]
Choose a partition of $\{1,\ldots,n\}$ into one block of size~$N$, exactly $k$ blocks of size~$2$, and the remaining blocks singletons.
Let $\U_n$ and $\widetilde{\U}_n$ be the corresponding semigroups from \cref{prop:lower-zero-free-strong,prop:lower-zero-free-signed-family}.

By \cref{prop:lower-zero-free-strong}, the semigroup $\U_n$ is strongly connected.
By \cref{prop:lower-zero-free-signed-family}, the semigroup $\widetilde{\U}_n$ is its signed version and is zero-free.
Hence, by \cref{prop:lower-zero-free-signed-irreducible}, the semigroup~$\widetilde{\U}_n$ is irreducible.

It remains to estimate its size.
For each of the $k$ blocks of size~$2$, the corresponding factor in the formula from \cref{prop:lower-zero-free-signed-family} is
\[
  \sum_{t=1}^m (2 \cdot 2)^{b_t}
  =
  \sum_{t=1}^m 4^{b_t}
  \ge
  4^N\,.
\]
All remaining factors are at least~$1$, so
\[
  |\widetilde{\U}_n|
  \ge
  (4^N)^k
  =
  2^{2 N k}
  =
  2^{2 \lfloor n/2 \rfloor \lfloor n/4 \rfloor}.
\]
This completes the proof.
\end{proof}

\section{Conclusions and open problems}\label{sec:conclusions}
Our $3^{n^2}$ bound on the cardinality of finite irreducible rational matrix semigroups (\cref{thm:main-bound}) breaks the barrier of $2^{\OO(n^2 \log n)}$ suggested in previous works \cite{Schutzenberger62,Berstel2011,AlmeidaSteinberg09,BumpusHKST20}.
Up to a constant in the exponent our bound is tight (\cref{prop:lower}).
As discussed in the introduction, the largest finite rational $n \times n$ matrix groups are known explicitly, using the classification of finite simple groups.
It would be similarly intriguing to identify the largest finite irreducible rational $n \times n$ matrix semigroups.
By our results, they have $2^{\Theta(n^2)}$ elements.

While we now have a good understanding of the maximal cardinality of finite irreducible matrix semigroups, for their diameter there is still a gap between the best known lower bound of $2^{n + \Theta(\sqrt{n \log n})}$ ~\cite{Panteleev15} and the upper bound of~$2^{\Theta(n^2)}$ implied by our work. As mentioned in the introduction, the gap for the mortality threshold is even bigger, since only polynomial lower bounds are known.

It would be also interesting to understand if the upper bound from~\cref{thm:main-bound} can be made more precise if it is also allowed to depend on the number of generators. The examples from our lower bounds have exponentially many generators. For transformation semigroups (equivalently, semigroups of matrices with entries in $\{0, 1\}$ with exactly one nonzero entry in every row) this question was studied in~\cite{Holzer2004}. The diameter of transformation semigroups with a bounded number of generators was studied in~\cite{Ryzhikov2025DiamArxiv}. The cardinality of aperiodic transformation semigroups was studied in~\cite{Brzozowski2015}.



\bibliographystyle{alphaurl}
\bibliography{references}

@article{Mandel1977,
  author       = {Arnaldo Mandel and
                  Imre Simon},
  title        = {On Finite Semigroups of Matrices},
  journal      = {Theoretical Computer Science},
  volume       = {5},
  number       = {2},
  pages        = {101--111},
  year         = {1977},
  doi          = {10.1016/0304-3975(77)90001-9},
}

@article{Ryzhikov2025DiamArxiv,
  author       = {Andrew Ryzhikov},
  title        = {Careful synchronisation and the diameter of transformation semigroups
                  with few generators},
  journal      = {CoRR},
  volume       = {abs/2506.13962},
  year         = {2025},
  url          = {https://doi.org/10.48550/arXiv.2506.13962},
  doi          = {10.48550/ARXIV.2506.13962},
  eprinttype    = {arXiv},
  eprint       = {2506.13962},
  bibsource    = {dblp computer science bibliography, https://dblp.org}
}

@article{Holzer2004,
  author       = {Markus Holzer and
                  Barbara K{\"{o}}nig},
  title        = {On deterministic finite automata and syntactic monoid size},
  journal      = {Theoretical Computer Science},
  volume       = {327},
  number       = {3},
  pages        = {319--347},
  year         = {2004},
  doi          = {10.1016/J.TCS.2004.04.010},
}

@article{Rystsov1997,
  author       = {Igor K. Rystsov},
  title        = {Reset Words for Commutative and Solvable Automata},
  journal      = {Theoretical Computer Science},
  volume       = {172},
  number       = {1-2},
  pages        = {273--279},
  year         = {1997},
  doi          = {10.1016/S0304-3975(96)00136-3},
}

@article{Rystsov1992Rank,
  author={Igor Rystsov},
  title={Rank of a finite automaton},
  journal={Cybernetics and Systems Analysis},
  volume={28},
  number={3},
  pages={323--328},
  year={1992},
  publisher={Kluwer Academic Publishers-Plenum Publishers New York}
}

@inproceedings{Colcombet2019,
  author       = {Thomas Colcombet and
                  Jo{\"{e}}l Ouaknine and
                  Pavel Semukhin and
                  James Worrell},
  editor       = {Christel Baier and
                  Ioannis Chatzigiannakis and
                  Paola Flocchini and
                  Stefano Leonardi},
  title        = {On Reachability Problems for Low-Dimensional Matrix Semigroups},
  booktitle    = {46th International Colloquium on Automata, Languages, and Programming,
                  {ICALP} 2019, July 9-12, 2019, Patras, Greece},
  series       = {LIPIcs},
  volume       = {132},
  pages        = {44:1--44:15},
  publisher    = {Schloss Dagstuhl - Leibniz-Zentrum f{\"{u}}r Informatik},
  year         = {2019},
  doi          = {10.4230/LIPICS.ICALP.2019.44},
}

@article{Almeida2009representation,
  title={Representation theory of finite semigroups, semigroup radicals and formal language theory},
  author={Almeida, Jorge and Margolis, Stuart and Steinberg, Benjamin and Volkov, Mikhail},
  journal={Transactions of the American Mathematical Society},
  volume={361},
  number={3},
  pages={1429--1461},
  year={2009}
}

@article{Brzozowski2015,
  author       = {Janusz A. Brzozowski and
                  Marek Szykula},
  title        = {Large Aperiodic Semigroups},
  journal      = {International Journal of Foundations of Computer Science},
  volume       = {26},
  number       = {7},
  pages        = {913--932},
  year         = {2015},
  url          = {https://doi.org/10.1142/S0129054115400067},
  doi          = {10.1142/S0129054115400067},
  bibsource    = {dblp computer science bibliography, https://dblp.org}
}

@inproceedings{Rida2025,
  author = {Rida {Ait El Manssour} and Roland Guttenberg and Nathan Lhote and Mahsa
Shirmohammadi and James Worrell},
  title = {Revisiting Finiteness of Matrix Monoids},
  editor       = {Sayan Bhattacharya and Danupon Nanongkai and
                  Michael Benedikt and Gabriele Puppis},
  booktitle    = {53rd International Colloquium on Automata, Languages, and
                  Programming, {ICALP} 2026},
  series       = {LIPIcs},
  volume       = {374},
  pages        = {158:1--158:17},
  publisher    = {Schloss Dagstuhl - Leibniz-Zentrum f{\"{u}}r Informatik},
  year         = {2026},
  doi          = {10.4230/LIPIcs.ICALP.2026.158},
}

@book{Steinberg2016,
  title={Representation theory of finite monoids},
  author={Steinberg, Benjamin},
  year={2016},
  publisher={Springer}
}

@book{Minc1988,
  title={Nonnegative matrices},
  author={Minc, Henryk},
  year={1988},
  publisher={Wiley-Interscience}
}

@inproceedings{Ryzhikov2024RP,
  author       = {Andrew Ryzhikov},
  editor       = {Laura Kov{\'{a}}cs and
                  Ana Sokolova},
  title        = {On Shortest Products for Nonnegative Matrix Mortality},
  booktitle    = {Reachability Problems -- 18th International Conference, {RP} 2024,
                  Vienna, Austria, September 25-27, 2024, Proceedings},
  series       = {Lecture Notes in Computer Science},
  volume       = {15050},
  pages        = {104--119},
  publisher    = {Springer},
  year         = {2024},
  url          = {https://doi.org/10.1007/978-3-031-72621-7\_8},
  doi          = {10.1007/978-3-031-72621-7\_8},
  bibsource    = {dblp computer science bibliography, https://dblp.org}
}

@inproceedings{Kiefer2025,
  author       = {Stefan Kiefer and
                  Andrew Ryzhikov},
  editor       = {Olaf Beyersdorff and
                  Michal Pilipczuk and
                  Elaine Pimentel and
                  Kim Thang Nguyen},
  title        = {Efficiently Computing the Minimum Rank of a Matrix in a Monoid of
                  Zero-One Matrices},
  booktitle    = {42nd International Symposium on Theoretical Aspects of Computer Science,
                  {STACS} 2025, March 4-7, 2025, Jena, Germany},
  series       = {LIPIcs},
  volume       = {327},
  pages        = {61:1--61:22},
  publisher    = {Schloss Dagstuhl - Leibniz-Zentrum f{\"{u}}r Informatik},
  year         = {2025},
  doi          = {10.4230/LIPICS.STACS.2025.61},
}

@article{Bell2021,
  author       = {Paul C. Bell and
                  Igor Potapov and
                  Pavel Semukhin},
  title        = {On the mortality problem: From multiplicative matrix equations to
                  linear recurrence sequences and beyond},
  journal      = {Information and Computation},
  volume       = {281},
  pages        = {104736},
  year         = {2021},
  doi          = {10.1016/J.IC.2021.104736},
}

@article{Cassaigne2014,
  author       = {Julien Cassaigne and
                  Vesa Halava and
                  Tero Harju and
                  Fran{\c{c}}ois Nicolas},
  title        = {Tighter Undecidability Bounds for Matrix Mortality, Zero-in-the-Corner
                  Problems, and More},
  journal      = {CoRR},
  volume       = {abs/1404.0644},
  year         = {2014},
  url          = {http://arxiv.org/abs/1404.0644},
  eprinttype    = {arXiv},
  eprint       = {1404.0644},
}

@article{Paterson1970,
  title={Unsolvability in $3 \times 3$ Matrices},
  author={Mike Paterson},
  journal={Studies in Applied Mathematics},
  year={1970},
  volume={49},
  pages={105--107}
}

@inproceedings{AlmeidaSteinberg09,
  author       = {Jorge Almeida and
                  Benjamin Steinberg},
  editor       = {Volker Diekert and
                  Dirk Nowotka},
  title        = {Matrix Mortality and the {{\v{C}}}ern{\'{y}}-{P}in Conjecture},
  booktitle    = {Developments in Language Theory, 13th International Conference, {DLT}
                  2009, Stuttgart, Germany, June 30 - July 3, 2009. Proceedings},
  series       = {Lecture Notes in Computer Science},
  volume       = {5583},
  pages        = {67--80},
  publisher    = {Springer},
  year         = {2009},
  url          = {https://doi.org/10.1007/978-3-642-02737-6\_5},
  doi          = {10.1007/978-3-642-02737-6\_5}
}

@article{Berry04,
  author={Berry, Neil and Dubickas, Artūras and Elkies, Noam D. and Poonen, Bjorn and Smyth, Chris},
  journal={Quarterly Journal of Mathematics}, 
  title={The conjugate dimension of algebraic numbers}, 
  year={2004},
  volume={55},
  number={3},
  pages={237-252},
}

@book{Berstel2011,
  title={Noncommutative rational series with applications},
  author={Berstel, Jean and Reutenauer, Christophe},
  number={137},
  year={2011},
  publisher={Cambridge University Press}
}

@inproceedings{BumpusHKST20,
  author       = {Georgina Bumpus and
                  Christoph Haase and
                  Stefan Kiefer and
                  Paul-Ioan Stoienescu and
                  Jonathan Tanner},
  editor       = {Artur Czumaj and
                  Anuj Dawar and
                  Emanuela Merelli},
  title        = {On the Size of Finite Rational Matrix Semigroups},
  booktitle    = {47th International Colloquium on Automata, Languages, and Programming,
                  {ICALP} 2020, July 8-11, 2020, Saarbr{\"{u}}cken, Germany (Virtual
                  Conference)},
  series       = {LIPIcs},
  volume       = {168},
  pages        = {115:1--115:13},
  publisher    = {Schloss Dagstuhl - Leibniz-Zentrum f{\"{u}}r Informatik},
  year         = {2020},
  doi          = {10.4230/LIPICS.ICALP.2020.115},
}

@unpublished{FeitUnp,
author = {Walter Feit},
title = {The orders of finite linear groups},
note = {Unpublished preprint}
}

@article{Friedland97,
author = {Shmuel Friedland},
title = {The Maximal Orders of Finite Subgroups in $\mathrm{GL}_n(\mathbb{Q})$},
journal = {Proceedings of the American Mathematical Society},
volume = {125},
number = {12},
year = {1997},
pages = {3519--3526}
}

@book{Howie1995,
  author    = {John M. Howie},
  title     = {Fundamentals of Semigroup Theory},
  publisher = {Clarendon Press},
  series    = {Oxford Science Publications},
  address   = {Oxford},
  year      = {1995},
  isbn      = {978-0-19-851194-9}
}

@article{KuzmanovichPavlichenkov02,
author = {James Kuzmanovich and Andrey Pavlichenkov},
title = {Finite Groups of Matrices Whose Entries Are Integers},
journal = {The American Mathematical Monthly},
volume = {109},
number = {2},
pages = {173--186},
year  = {2002}
}

@article{Minkowski87,
title = {Zur {T}heorie der positiven quadratischen {F}ormen},
author = {Hermann Minkowski},
pages = {196--202},
volume = {101},
journal = {Journal f\"{u}r die reine und angewandte Mathematik},
doi = {doi:10.1515/crll.1887.101.196},
year = {1887}
}

@book{Newman1972,
  author    = {Morris Newman},
  title     = {Integral Matrices},
  series    = {Pure and Applied Mathematics},
  volume    = {45},
  publisher = {Academic Press},
  address   = {New York and London},
  year      = {1972},
  isbn      = {0125178506}
}

@inproceedings{Panteleev15,
  author       = {Pavel Panteleev},
  editor       = {Adrian-Horia Dediu and
                  Enrico Formenti and
                  Carlos Mart{\'{\i}}n-Vide and
                  Bianca Truthe},
  title        = {Preset Distinguishing Sequences and Diameter of Transformation Semigroups},
  booktitle    = {Language and Automata Theory and Applications - 9th International
                  Conference, {LATA} 2015, Nice, France, March 2-6, 2015, Proceedings},
  series       = {Lecture Notes in Computer Science},
  volume       = {8977},
  pages        = {353--364},
  publisher    = {Springer},
  year         = {2015},
  url          = {https://doi.org/10.1007/978-3-319-15579-1\_27},
  doi          = {10.1007/978-3-319-15579-1\_27},
  bibsource    = {dblp computer science bibliography, https://dblp.org}
}

@article{Schutzenberger62,
  author       = {Marcel Paul Sch{\"{u}}tzenberger},
  title        = {Finite Counting Automata},
  journal      = {Information and Control},
  volume       = {5},
  number       = {2},
  pages        = {91--107},
  year         = {1962},
  url          = {https://doi.org/10.1016/S0019-9958(62)90244-9},
  doi          = {10.1016/S0019-9958(62)90244-9},
  bibsource    = {dblp computer science bibliography, https://dblp.org}
}

@article{Weisfeiler84,
author = {Boris Weisfeiler},
title = {Post-classification version of {J}ordan's theorem on finite linear groups},
journal = {Proceedings of the National Academy of Sciences of the United States of America},
volume = {81},
pages = {5278--5279},
year  = {1984}
}

@unpublished{WeisfeilerUnp,
author = {Boris Weisfeiler},
title = {On the Size and Structure of Finite Linear Groups},
note = {Unpublished preprint, see \url{https://arxiv.org/abs/1203.1960}}
}

@inproceedings{KieferRyzhikovSTACS26,
  author       = {Stefan Kiefer and
                  Andrew Ryzhikov},
  editor       = {Meena Mahajan and
                  Florin Manea and
                  Annabelle McIver and
                  Kim Thang Nguyen},
  title        = {The Asymptotic Size of Finite Irreducible Semigroups of Rational Matrices},
  booktitle    = {43rd International Symposium on Theoretical Aspects of Computer Science,
                  {STACS} 2026, Grenoble, France, March 9-13, 2026},
  series       = {LIPIcs},
  pages        = {60:1--60:20},
  publisher    = {Schloss Dagstuhl - Leibniz-Zentrum f{\"{u}}r Informatik},
  year         = {2026},
  doi          = {10.4230/LIPICS.STACS.2026.60},
}

@article{Steinberg26,
  author       = {Benjamin Steinberg},
  title        = {A short proof of a bound on the size of finite irreducible semigroups
                  of rational matrices},
  journal      = {CoRR},
  volume       = {abs/2601.03206},
  year         = {2026},
  url          = {https://doi.org/10.48550/arXiv.2601.03206},
  doi          = {10.48550/ARXIV.2601.03206},
  eprinttype   = {arXiv},
  eprint       = {2601.03206},
  bibsource    = {dblp computer science bibliography, https://dblp.org}
}

\appendix
\section{Semigroup Theory} \label{sec:background-semigroup}

We draw on some background from semigroup theory, and follow \cite{Howie1995}. We denote the zero of a semigroup by $0$.
A semigroup without zero is called \emph{simple} if it has no proper ideals.
A semigroup~$S$ with zero is called \emph{0-simple} if
\begin{enumerate}
\item[(i)] its only ideals are the zero ideal $\{0\}$ and~$S$;
\item[(ii)] $S^2 \ne \{0\}$.
\end{enumerate}

A \emph{minimal} ideal of a semigroup is an ideal that is minimal within the set of all ideals.
A \emph{0-minimal} ideal of a semigroup with zero is an ideal that is minimal within the set of all nonzero ideals.
Every finite semigroup has a minimal ideal.
Every finite semigroup with zero $S \ne \{0\}$ also has a 0-minimal ideal.
We have the following proposition.
\begin{proposition}[{\cite[Proposition~3.1.3]{Howie1995}}] \label{prop:3.1.3}
If $T$~is a \mbox{(0-)}minimal ideal of a semigroup then either $T^2 = \{0\}$ or $T$~is a \mbox{(0-)}simple semigroup.
\end{proposition}
An idempotent~$e \ne 0$ is called \emph{primitive} if for every idempotent $f$
\[
 e f  =  f e  =  f  \ne  0 \quad \text{implies} \quad e  =  f\,.
\]
A semigroup is called \emph{completely} \mbox{(0-)}simple if it is \mbox{(0-)}simple and has a primitive idempotent.
We have the following proposition.
\begin{proposition}[{\cite[Proposition~3.2.1]{Howie1995}}] \label{prop:3.2.1}
Every finite \mbox{(0-)}simple semigroup is completely \mbox{(0-)}simple.
\end{proposition}
In particular, every finite \mbox{(0-)}simple semigroup has a primitive nonzero idempotent.

A \emph{Rees 0-matrix semigroup} $M^0[G; I, \Lambda; P]$ is a semigroup with zero $(I \times G \times \Lambda) \cup \{0\}$, where $G$~is a group, $I, \Lambda$ are nonempty index sets, and the ``sandwich'' matrix $P = (p_{\lambda i})$ is a $\Lambda \times I$ matrix with entries in the 0-group $G^0 (= G \cup \{0\})$ such that no row or column of~$P$ consists entirely of zeros.
Multiplication is defined by
\begin{align*}
  (i, a, \lambda) (j, b, \mu) \ &= \ \begin{cases} (i,a p_{\lambda j} b, \mu) & \text{ if } p_{\lambda j} \ne 0, \\
                                                   0                          & \text{ if } p_{\lambda j} = 0,
                                     \end{cases} \\
  (i,a,\lambda) 0 \ &= \ 0 (i,a,\lambda) \ = \ 0 0 \ = \ 0\,.   
\end{align*}
A \emph{Rees matrix semigroup} $M[G; I, \Lambda; P]$ is analogously defined without zero.
The following theorem is an important result of semigroup theory: one can represent completely \mbox{(0-)}simple semigroups as Rees \mbox{(0-)}matrix semigroups.
\begin{theorem}[The Rees-Suschkewitsch Theorem {\cite[Theorems 3.2.3 and 3.3.1]{Howie1995}}] \label{thm:Rees}
Every Rees \mbox{(0-)}matrix semigroup is completely \mbox{(0-)}simple, and every completely \mbox{(0-)}simple semigroup is isomorphic to a Rees \mbox{(0-)}matrix semigroup.
\end{theorem}

We use the Rees-Suschkewitsch theorem to derive the following lemma.
\begin{lemma} \label{lem:group}
Let $S$ be a completely \mbox{(0-)}simple semigroup, and let $e \in S \setminus \{0\}$ be an idempotent.
Then $e S e \setminus \{0\}$ is a group with identity~$e$.
\end{lemma}
\begin{proof}
By \cref{thm:Rees} we can assume without loss of generality that $S$~is a Rees \mbox{(0-)}matrix semigroup $M[G; I, \Lambda; P]$ (or $M^0[G; I, \Lambda; P]$).
Let $e \ne 0$ be a nonzero idempotent.
Write $(i, p^{-1}, \lambda) \coloneqq e$.
Since $e^2 = e \ne 0$, we must have $p_{\lambda i} \ne 0$.
Thus, $(i, p^{-1}, \lambda)^2 = (i,p^{-1} p_{\lambda i} p^{-1}, \lambda)$, and so $p = p_{\lambda i} \ne 0$.

We show that
\begin{equation}
e S e \setminus \{0\} \ = \ \{(i, g, \lambda) \mid g \in G\}\,. \label{eq:group}
\end{equation}
Indeed, the inclusion $\mathord{\subseteq}$ follows from the definition of multiplication in a Rees \mbox{(0-)}matrix semigroup, since any nonzero product $e (j, h, \mu) e$ has the form $(i, g, \lambda)$ for some $g \in G$.
Towards the reverse inclusion, since the $\lambda$-row and $i$-column of~$P$ are not all zeros, there are $j$ with $p_{\lambda j}\neq 0$ and $\mu$ with $p_{\mu i}\neq 0$.
For any $g\in G$, set
\[
h \ \coloneqq \ p_{\lambda j}^{-1} p g p p_{\mu i}^{-1}.
\]
Then
\[
e (j,h,\mu) e
\ = \ (i, p^{-1} p_{\lambda j} h p_{\mu i} p^{-1}, \lambda)
\ = \ (i, g, \lambda)\,.
\]
This proves~\cref{eq:group}.

It follows from~\cref{eq:group} that $e S e \setminus \{0\}$ is a subsemigroup of~$S$; indeed, since $p = p_{\lambda i} \ne 0$, we have $(i,g,\lambda) (i,h,\lambda) = (i,g p h,\lambda) \in e S e \setminus \{0\}$.

Define
\[
 \phi : e S e \setminus \{0\} \to G  \qquad \text{with} \qquad \phi((i, g, \lambda)) \ \coloneqq \ g p\,.
\]
Then $\phi$~is injective (right-cancellation in~$G$).
It follows from~\cref{eq:group} that $\phi$~is surjective; indeed, for any $g \in G$ we have $(i,g p^{-1},\lambda) \in e S e \setminus \{0\}$ and
$\phi((i,g p^{-1},\lambda))=g$.
In fact, $\phi$~is an isomorphism; indeed,
\begin{align*}
 \phi((i,g,\lambda) (i,h,\lambda)) \ &= \ \phi((i, g p h, \lambda)) \ = \ g p h p \ = \  \phi((i,g,\lambda)) \phi((i,h,\lambda)) \qquad \text{and}\\
 \phi(e) \ &= \ \phi((i,p^{-1},\lambda)) \ = \ p^{-1} p \ = \ 1\,.
\end{align*}
It follows that $e S e \setminus \{0\}$ is a group with identity~$e$.
\end{proof}


\section{Proofs from \texorpdfstring{\Cref{sec:upper-two}}{Section~\ref{sec:upper-two}}}

\SKT{The following lemma from basic group theory is used in the proofs of \cref{lem:bound-on-group-size,lem:Hb-bound}.}
\begin{lemma} \label{lem:group-ker}
Let $G$ be a group, $H$ a semigroup, and let $\varphi: G \to H$ be a surjective homomorphism.
Set $e_H \coloneqq \varphi(1_G)$.
Then $H$~is a group with identity~$e_H$.
Moreover, if $\{g \in G \mid \varphi(g) = e_H\} = \{1_G\}$, then $\varphi$~is a group isomorphism.
\end{lemma}
\begin{proof}[Proof of \Cref{lem:group-ker}]
For any $g \in G$, we have $\varphi(g) \varphi(g^{-1}) = \varphi(g g^{-1}) = \varphi(1_G) = e_H$ and $\varphi(g^{-1}) \varphi(g) = e_H$.
Since $\varphi$~is surjective, it follows that every element of~$H$ has a two-sided inverse.
Also, for any $g\in G$,
\[
e_H \varphi(g) \ = \ \varphi(1_G) \varphi(g) \ = \ \varphi(1_G g) \ = \ \varphi(g) \qquad \text{and} \qquad
\varphi(g) e_H \ = \ \varphi(g 1_G) \ = \ \varphi(g)\,,
\]
so $e_H$~is a two-sided identity.
Hence $H$ is a group.

Assume that $\{g \in G \mid \varphi(g) = e_H\} = \{1_G\}$.
To show that $\varphi$~is injective, let $\varphi(g_1) = \varphi(g_2)$.
Then
\[
 \varphi(g_1 g_2^{-1}) \ = \ \varphi(g_1) \varphi(g_2^{-1}) \ = \ \varphi(g_2) \varphi(g_2^{-1}) \ = \ \varphi(g_2 g_2^{-1}) \ = \ \varphi(1_G) \ = \ e_H\,.
\]
By the assumption, $g_1 g_2^{-1} = 1_G$.
Thus, $g_1 = g_2$.
It follows that $\varphi$~is injective.
Hence, it is a group isomorphism.
\end{proof}

\begin{proof}[Proof of \cref{lem:bound-on-group-size}]
By a folklore result, any finite subgroup of~$\GL_n(\Q)$ is conjugate to a finite subgroup of~$\GL_n(\Z)$; see, e.g., \cite[Theorem~1.6]{KuzmanovichPavlichenkov02}.
So it suffices to show the lemma for integer matrix groups.
Let $G$ be a finite subgroup of~$\GL_n(\Z)$.

Fix an \emph{odd} prime $p \ge 3$ and write $\nu_p : \Z^{n \times n} \to (\Z/p \Z)^{n \times n}$ for the map that reduces each entry mod~$p$.
Minkowski~\cite{Minkowski87} (cf.~\cite[Proposition~1.3]{KuzmanovichPavlichenkov02}) proved in~1887 that for any matrix $X \in \Z^{n \times n}$ with $\nu_p(X) = I_n$ (where $I_n$~is the $n \times n$ identity matrix), if there is a prime~$q$ with $X^q = I_n$ then $X = I_n$.

We use \cref{lem:group-ker} to show that the restriction of~$\nu_p$ to~$G$,
\[
 \nu_p|_G: G \to \nu_p(G)\,,
\]
is a group isomorphism from~$G$ to $\nu_p(G)$.
Towards a contradiction, suppose that there is $g \in G \setminus \{I_n\}$ with $\nu_p(g) = I_n$.
Since $G$~is finite, the order of~$g$, say~$m$, is finite, and we have $g^m = I_n$ with $m > 1$.
Let $q$~be a prime factor of~$m$.
Then $h \coloneqq g^{m/q}$ has order~$q$, i.e., $h^q = I_n$.
Since $\nu_p$~is a homomorphism, $\nu_p(h) = \nu_p(g^{m/q}) = \nu_p(g)^{m/q} = I_n^{m/q} = I_n$.
Minkowski's result mentioned above implies $h = I_n$, contradicting that $h$~has order $q \ge 2$.
Thus, $I_n$~is the only matrix from~$G$ that $\nu_p$~maps to~$I_n$.
Hence, by \cref{lem:group-ker}, $G$~is isomorphic to $\nu_p(G)$.

Take $p=3$.
Since $\nu_3(G) \subseteq (\Z/3\Z)^{n \times n} \setminus \{O_n\}$, we have $|G| = |\nu_3(G)| \le 3^{n^2}-1$.
\end{proof}

\end{document}